\title{Lexicographic products and the power of non-linear network coding}
\author{ {Anna Blasiak\thanks{Department of Computer Science, Cornell University, Ithaca NY 14853. E-mail: {\tt ablasiak@cs.cornell.edu}. Supported by an NDSEG Graduate Fellowship, an AT\&T
Labs Graduate Fellowship, and an NSF Graduate Fellowship.}} \qquad
{Robert Kleinberg\thanks{Department of Computer Science, Cornell University, Ithaca NY 14853. E-mail:
{\tt rdk@cs.cornell.edu}. Supported in part by NSF grant CCF-0729102, AFOSR grant FA9550-09-1-0100, a Microsoft Research New
Faculty Fellowship, a Google Research Grant, and an Alfred P. Sloan Foundation Fellowship.}} \qquad
{Eyal Lubetzky\thanks{Microsoft
Research, One Microsoft Way, Redmond, WA 98052, USA. Email:
{\tt eyal@microsoft.com}.}}}

\documentclass [11pt]{article}

\usepackage[]{amsmath,amssymb,amsfonts,latexsym,amsthm,enumerate,paralist,url}
\usepackage{amssymb,graphicx}
\usepackage{booktabs}
\usepackage[numeric,initials,nobysame]{amsrefs}
\usepackage[compact]{titlesec}
\usepackage{xspace}

\date{}

\numberwithin{equation}{section}

\newtheorem{theorem}{Theorem}[section]
\newtheorem*{theorem*}{Theorem}
\newtheorem{lemma}[theorem]{Lemma}
\newtheorem{claim}[theorem]{Claim}
\newtheorem{proposition}[theorem]{Proposition}

\newtheorem*{observation*}{Observation}

\theoremstyle{definition}{
\newtheorem{example}[theorem]{Example}
\newtheorem{definition}[theorem]{Definition}

\newtheorem*{remark*}{Remark}
}

\newcommand{\deq}{\stackrel{\scriptscriptstyle\triangle}{=}}

\newcommand{\R}{\mathbb R}

\newcommand{\N}{\mathbb N}

\newcommand{\F}{\mathbb F}

\renewcommand{\epsilon}{\varepsilon}

\newcommand{\chibar}{\overline{\chi}}

\titlespacing{\subsection}{0pt}{*1}{*1}
\titlespacing{\section}{0pt}{*1}{*1}

\newenvironment{lparray}%
{\begin{array}{l@{\hspace{4mm}}l@{\hspace{5mm}}l}}%
{\end{array}}
\newlength{\lplb}
\newcommand{\is}{I}   
\newcommand{\otheris}{J}  
\newcommand{\powerset}[1]{{\mathcal{P}(#1)}}

\newcommand{\cm}{A}
\newcommand{\cset}{\mathcal{Q}}
\newcommand{\clos}{\operatorname{cl}}
\newcommand{\dualcoeff}{{d}}
\newcommand{\trans}{{\mathsf{T}}}
\newcommand{\ones}{{\mathbf{1}}}
\newcommand{\onesi}[1][i]{{\ones}_{#1}}
\newcommand{\lexp}{\bullet}
\newcommand{\sbv}{e}
\newcommand{\xg}{{\xi^G}}
\newcommand{\yg}{{\eta^G}}
\newcommand{\xf}{{\xi^F}}
\newcommand{\yf}{{\eta^F}}
\newcommand{\xgf}{{\xi^{G \lexp F}}}
\newcommand{\ygf}{{\eta^{G \lexp F}}}
\newcommand{\lpv}{\rho}
\newcommand{\hstx}{\hat{X}}
\newcommand{\hsty}{\hat{Y}}
\newcommand{\lpB}{\mathfrak{B}}

\newcommand{\eps}{\varepsilon}

\newcommand{\chr}{{\mathrm{char}}}

\newcommand{\groundset}{{\mathcal U}}
\newcommand{\odds}{{\mathcal O}}
\newcommand{\oddsm}{{\mathcal B}}

\newcommand{\linrate}{\lambda} 
\newcommand{\Encode}{\mathcal{E}}
\newcommand{\bsf}{{b^{\F}}}

\newcommand{\fano}{\mathcal{F}}
\newcommand{\nonfano}{\mathcal{N}}
\newcommand{\dimvec}{\vec{\mathbf{d}}}
\newcommand{\ds}{\dim}
\newcommand{\dds}{\Delta \! \ds}

\newcommand{\rankvec}{\vec{\mathbf{r}}}

\newcommand{\ABnote}[1]{}
\newcommand{\BKnote}[1]{}

\begin{document}
\maketitle
\begin{abstract}
We introduce a technique for establishing and amplifying gaps
between parameters of network coding and index coding problems.  The technique
uses linear programs to
establish separations between combinatorial
and coding-theoretic parameters and applies hypergraph lexicographic
products to amplify these separations.  This entails combining the
dual solutions of the lexicographic multiplicands and proving that
this is a valid dual solution of the product.  Our result is general enough to apply to
a large family of linear programs.  This blend of linear
programs and lexicographic products gives a recipe for constructing
hard instances in which the gap between
combinatorial or coding-theoretic parameters is \emph{polynomially
  large}.  We find polynomial gaps in cases in which the largest previously
known gaps were only small constant factors or entirely unknown.  Most notably, we
show a polynomial separation between linear and non-linear
network coding rates. This involves exploiting a connection between
matroids and index coding to establish a previously unknown separation
between linear and non-linear index coding rates.  We also construct
index coding problems with a polynomial gap between the broadcast rate
and the trivial lower bound for which no gap was previously known.
\end{abstract}

\section{Introduction}
 The problem of \emph{Network Coding}, introduced by Ahlswede
\emph{et al}~\cite{ACLY} in 2000, asks for the maximum rate
at which information can be passed from a set of sources to a set of
targets in a capacitated network.
%
In practice, there are many examples where network coding provides
faster transmission rates compared to traditional routing, e.g.~\cite{KRHKMC} details a recent
 one in wireless networks.  However, despite tremendous initial
 success in using network coding to solve some \emph{broadcast}
 problems (those in which every receiver demands the same message),
 very little is known about how to
 compute or approximate the network coding rate in general.
(See~\cite{YLC} for a survey of the topic.)


In the absence of general algorithms for solving network coding,
attention has naturally turned to restricted models of coding (e.g.\ linear functions
between vector spaces over finite fields) and to approximating network
coding rates using graph-theoretic parameters (e.g.\  minimum
cut~\cite{AHJKL} and the independence number~\cite{AHLSW}).
Several of these variants provide bounds on the network coding
rate, but the worst-case approximation factor of these bounds remains
unknown. For example, it is known that there exists a network
in which non-linear network coding can achieve a rate
which exceeds the best linear
network code by a factor of
$\frac{11}{10}$~\cite{DFZ1}, but it is not known whether
this gap\footnote{The literature on
network coding distinguishes between~\emph{linear} network
codes, in which the messages are required to be elements
of a finite field, and \emph{vector-linear} network codes,
in which the messages are elements of a finite-dimensional
vector space over a finite field.  Linear coding is weaker,
and a gap of size $n^{1-\epsilon}$ is known~\cite{LuSt}.
Vector-linear coding is much more powerful, and no gap larger
than 11/10 was known prior to our work.} can
be improved to $n^{1-\epsilon}$, or even possibly
to $\Theta(n)$.

In this paper we
introduce a general technique for amplifying
many of these gaps by combining linear programming with
hypergraph product operations.  For instance, this enables
us to
construct a family of network coding instances with
$n$ messages, in which the rate of the best non-linear
network code exceeds the rate of the best (vector-)linear network
code by a factor of at least $n^{\epsilon}$.  A crucial ingredient
in our technique is \emph{index coding}~\cites{BK,BBJK},
a class of communication
problems in which a server holds a set of messages that it
wishes to broadcast over a
noiseless channel to a set of receivers. Each receiver is interested
in one of the messages and has side-information comprising some subset
of the other messages.  The
objective is to devise an optimal encoding scheme 
(one minimizing the broadcast length) that allows all the
receivers to retrieve their required information.
Following~\cite{AHLSW}, we use $\beta$ to denote the
limiting value of the information
rate (i.e., ratio of broadcast length to message length) of this
optimal scheme, as the message length tends to infinity.

In our framework, index coding is most useful
for isolating a sub-class of network coding
problems that can be combined using lexicographic products.
However, it is also an important and well-studied problem
in its own right.
Index coding is intimately related to network coding in general.
It is essentially equivalent to
the special case of network coding in which
only one edge has finite capacity.\footnote{The unique
finite-capacity edge
represents the broadcast channel.  Each sender is connected to
the tail of this edge, each receiver is connected to its head,
and each receiver has incoming edges directly from a subset of
the senders, representing the side-information.}
Additionally,~\cite{RSG} shows that
linear network coding can be reduced to linear index coding, thus
implying that index coding captures much of the difficulty of
network coding.

Index coding is also intricately
related to other well-studied areas of mathematics.
Connections between matroids and index coding were
established in~\cite{RSG2}; for example, that paper
shows that realizability
of a matroid over a field $\F$ is
equivalent to linear solvability of a corresponding
index coding problem.
Index coding is also closely connected to graph theory:
a special case of index coding can be described by an
undirected graph $G$, representing a communication
problem where a broadcast channel
communicates messages to a set of vertices, each of
whom has side-information consisting of the neighbors'
messages.
Letting $\alpha(G),\chibar(G)$ denote the independence and clique-cover numbers of $G$, respectively, one has
 \begin{equation}
   \label{eq-trivial-ineqs}
   \alpha(G) \leq \beta(G) \leq \chibar(G)\,.
 \end{equation}
The first inequality above is due to an independent set being
identified with a set of receivers with no mutual information, whereas
the last one due to~\cites{BK,BBJK} is obtained by broadcasting the
bitwise {\sc xor} of the vertices per clique in the optimal clique-cover of
$G$.  As one consequence of the general technique we develop here,
we settle an open question of~\cite{AHLSW} by proving that
$\alpha(G)$ can differ from $\beta(G)$; indeed, we show
that their ratio can be as large as $n^{0.139}$.

\subsection{Contributions}

We present a general technique that amplifies
lower bounds for index coding problems
using lexicographic hypergraph products in
conjunction with linear programs that express
information-theoretic inequalities.  The use
of such linear programs
to prove lower bounds in network coding theory is
not new,
but, perhaps surprisingly, they have not gained
widespread use in the analysis of index coding problems.
We give an information-theoretic linear program, whose solution, $b$, gives the best known lower bound on $\beta$.
However, our main innovation is the insight that this
linear programming technique can be combined with
the combinatorial technique of graph products to
yield lower bounds for sequences of index coding and network coding problems.
Specifically, we provide a lexicographic product
operation on index coding problems along with
an operation that combines
dual solutions of the corresponding two linear programs.  We show that
the combined dual yields a dual solution of the linear program corresponding
to the lexicographic product.  Using this operation,
we demonstrate that index coding lower bounds proven using
linear programming behave supermultiplicatively
under lexicographic products.  This technical tool
enables us to prove some new separation results answering
open questions in the field.

Our technique not only applies to the standard linear
programs used in network information theory (those that
express entropy inequalities such as submodularity)
but to \emph{any} family of linear programs
constructed using what we call a \emph{tight homomorphic
constraint schema}.  In particular, if one can develop a tight homomorphic constraint schema that
applies to a restricted class of codes (e.g.\ linear) then it becomes
possible to prove lower bounds for this class of codes and amplify them
using lexicographic products.  We pursue this approach
in establishing a large multiplicative
gap between linear and non-linear network coding.

\begin{theorem}
Lower bounds for index coding problems can be proven by solving
a linear program whose constraints
are valid for the class of coding functions being considered.
If the linear program is constructed using a
 tight homomorphic constraint schema
 (see Section \ref{sec:lp}),
then its optimum is supermultiplicative under the lexicographic
product of two index coding problems.
\end{theorem}

To separate linear from non-linear coding,
we first produce a pair of linear inequalities
that are valid information inequalities for tuples
of random variables defined by linear functions
over fields of odd (resp., even) characteristic,
but not vice-versa.  We obtain these inequalities by
considering the Fano and non-Fano matroids; the former
is a matroid that is only realizable in characteristic 2,
while the latter is only realizable in odd characteristic
and in characteristic 0.
For each of the two matroids, we are able to transform
a proof of its non-realizability into a much stronger
quantitative statement about dimensions of
vector spaces over a finite field.  This, in turn,
we transform into a tight homomorphic constraint schema of
valid information inequalities for linear random
variables.

We then use the connection between matroids and index
coding~\cites{DFZ1,DFZ2,RSG2} and these inequalities
to give a pair of index coding instances where the best non-linear coding rate is
strictly better than the best linear rate over a field of odd
(resp.,even) characteristic.  We do this by establishing a general
theorem that says that for a matroid $M$, and an inequality that is
violated for the rank function of $M$, there is an index coding problem for
which the bound obtained by adding this inequality to the LP is
strictly greater than $b$.

We can now plug the constraint schema into our lexicographic
product technique and apply it to these two index coding problems to
yield the aforementioned separation between
(vector-)linear and non-linear network coding.

\begin{theorem}
There exists an explicit family of network coding instances (based on index coding instances)
with $n$ messages and some fixed $\epsilon > 0$ such that the non-linear rate is
$\Omega(n^\epsilon) $ times larger than the linear rate.
\label{thm:NWCgap}
\end{theorem}
The largest previously known gap between the non-linear and linear rates for network coding was a
factor of $\frac{11}{10}$ (\cite{DFZ2}). No separation was known between these parameters for index coding
(see~\cites{LuSt,AHLSW} for related separation results focusing on the weaker setting of scalar linear codes).

As explained above, given any index coding
problem $G$ we can write down an LP 
whose constraints are based on information inequalities
that gives a lower bound on $\beta$.  It is the best
known lower bound, and in many cases, strictly better than any
previously known bound.  Notably, we can show
that the broadcast rate of the 5-cycle is at least $\frac52$,
giving the first known gap between
the independence number $\alpha$ (which equals 2 for the
5-cycle) and
the broadcast rate
$\beta$.  Amplifying this gap using lexicographic
products, we can boost the ratio $\beta/\alpha$ to grow polynomially with $n$
in a family of $n$-vertex graphs.

\begin{theorem}
There exists an explicit family of index coding instances with $n$
messages such that $\beta(G)$ is at least $\Omega(n^{\delta})$ times
larger than $\alpha(G)$, where $\delta=1-2\log_5(2)\approx 0.139$.
\label{thm:alpha-beta-gap}
\end{theorem}

The remainder of the paper is organized as follows.  In Section~\ref{sec:defs}
we give a formal definition of index coding and the
lexicographic product of two index coding problems.  In Section~\ref{sec:lp}
we describe a general class of LPs and prove they behave
supermultiplicatively under lexicographic
products.  Section~\ref{sec:alpha-beta} is devoted to the proof of
Theorem~\ref{thm:alpha-beta-gap}.  In Section~\ref{sec:matroids} we give a
construction from matroids to index coding and prove a number of
connections between properties of the matroid and the parameters of
the corresponding index coding problem.  Finally, in Section~\ref{sec:nonlinear-linear-gap}
we establish inequalities that are valid for
linear codes over fields of odd (resp., even) characteristic and then use
these to prove Theorem~\ref{thm:NWCgap}.

\section{Definitions}
\label{sec:defs}

An index coding problem is specified by a \emph{directed
hypergraph} $G = (V,E)$, where elements of $V$ are
thought of as messages, and $E \subseteq V \times 2^V$
is a set of directed hyperedges $(v,S)$, each of which
is interpreted as a receiver who already knows the
messages in set $S$ and wants to receive message $v$.
Messages are drawn from a finite alphabet $\Sigma$,
and a solution of the problem specifies a finite alphabet $\Sigma_P$
to be used by the public channel, together with
an encoding scheme
$\Encode: \Sigma^{|V|} \to \Sigma_P$ such that,
for any possible values of $(x_v)_{v \in V}$,
every receiver $(v,S)$ is able to decode
the message $x_v$ from the value of $\Encode(\vec{x})$
together with that receiver's side information.
The minimum encoding length
$\ell = \left\lceil \log_{2} |\Sigma_P|\right\rceil$ for
messages that are $t$ bits long (i.e.~$\Sigma=\{0,1\}^t$)
is denoted by $\beta_t(G)$.
As noted in \cite{LuSt}, due to the  overhead associated
with relaying the side-information map to the server the
main focus is on the case $t\gg1$ and namely on the
following \emph{broadcast rate}.
\begin{equation}
  \label{eq-beta-limit}
  \beta(G) \deq \lim_{t\to\infty}\frac{\beta_t(G)}t = \inf_t \frac{\beta_t(G)}{t}
\end{equation}
(The limit exists by subadditivity.)
This is interpreted as the
average asymptotic number of broadcast bits needed per bit of input,
that is, the asymptotic broadcast rate for long messages.  We are also
interested in the optimal rate when we require that
$\Sigma$ is a finite-dimensional vector space over a finite
field $\F$, and the encoding function
is linear.  We denote this by $\linrate^\F$, and
we denote the optimal linear rate over any field as $\linrate$.

A useful notion in index coding is the following \emph{closure}
operation with respect to $G$, a given instance of the problem: for a set of messages $S \subseteq V$, define
\begin{equation}
  \label{eq-def-closure}
  \clos(S) = \clos_G(S) = S \cup \{x \mid \exists (x,T) \in E \mbox{ s.t. } T \subseteq S\}\,.
\end{equation}
The interpretation is that every message $x \in \clos(S)$ can
be decoded by someone who knows all of the messages in $S$ in
addition to the broadcast message.  In Section~\ref{sec:matroids}
when we discuss a transformation that associates an index
coding problem to every matroid, the closure operation
defined in this paragraph --- when specialized to the
index coding problems resulting from that transformation ---
will coincide with the usual matroid-theoretic closure operation.

We next define the lexicographic product operation for
directed hypergraphs, then proceed to present
Theorem~\ref{thm:beta-submult} which demonstrates
its merit in the context of index coding by showing that $\beta$ is
submultiplicative for this operation.  The proof gives further
intuition for the product operation.

\begin{definition}\label{def-lexicographic-prod}
The lexicographic product of two directed hypergraphs $G,F$, denoted
by $G \lexp F$, is a directed hypergraph whose vertex set is the
cartesian product $V(G) \times V(F)$.   The edge set of
$G \lexp F$ contains a directed hyperedge
$e$ for every pair of hyperedges $(e_G,e_F) \in E(G) \times E(F)$.
If $e_G = (w_G,S_G)$ and $e_F = (w_F, S_F)$, then the head of
$e = (e_G,e_F)$ is the ordered pair $(w_G,w_F)$ and its tail
is the set $(S_G \times V(F)) \cup (\{w_G\} \times S_F)$.
Denote by
$G^{\lexp n}$ the $n$-fold lexicographic power of $G$.
\end{definition}

\begin{remark*}
In the special case where the index coding problem is defined by a graph\footnote{When there are $n$ messages and exactly $n$ receivers,
w.l.o.g.\ receiver $i$ wants the message $x_i$ and one can encode the side-information by a graph on $n$ vertices which contains the edge $(i,j)$ iff receiver $i$ knows the message $x_j$.} the above definition coincides with the usual lexicographic graph product
(where $G\lexp F$ has the vertex set $V(G) \times V(F)$ and an edge from $(u,v)$ to $(u',v')$
iff either $(u,u')\in E(G)$ or $u=u'$ and $(v,v')\in E(F)$).
\end{remark*}

\begin{theorem}\label{thm:beta-submult}
The broadcast rate is submultiplicative under the lexicographic product of index coding problems. That is,
$\beta(G \lexp F) \le \beta(G) \, \beta(F)$ for any two directed hypergraphs $G$ and $F$.
\end{theorem}
\begin{proof}
Let $\epsilon > 0$ and, recalling the definition of $\beta$ in~\eqref{eq-beta-limit}
as the limit of $\beta_t/t$, let $K$ be a sufficiently large integer
such that for all $t \geq K$ we have $\beta_t(G)/t \leq \beta(G) + \epsilon$ as well as $\beta_t(F)/t \leq \beta(F) + \epsilon$.
Let $\Sigma = \{0,1\}^K$ and consider the following scheme for the index coding problem on $G \lexp F$
with input alphabet $\Sigma$, which will consist of an
inner and an outer code.

Let $\Encode_F$ denote an encoding function for $F$
with input alphabet $\Sigma$ achieving an optimal rate, i.e.\ minimizing $\log(|\Sigma_P|)/\log(|\Sigma|)$.
For each $v \in V(G)$, the inner code applies $\Encode_F$ to
the $|V(F)|$-tuple of messages
indexed by the set $\{v\} \times V(F)$,
obtaining a message $m_v$. 
Note that our assumption on $|\Sigma|$ implies that
the length of $m_v$ is equal to $K'$ for some
integer $K'$ such that $K \leq K' \leq (\beta(F)+\eps)K$.
Next, let $\Encode_G$ denote an optimal encoding function for $G$ with input $\{0,1\}^{K'}$.
The outer code applies $\Encode_G$ to $\{m_v\}_{v \in V(G)}$ and the assumption on $K$ ensures its output is at most
$(\beta(G)+\epsilon)K'$ bits long.

To verify that the scheme is a valid index code, consider a receiver in $G \lexp F$ represented by $e = ((w_G,w_F), (S_G \times V(F)) \cup (\{w_G\} \times S_F))$.
To decode $(w_G,w_F)$, the receiver first
computes $m_v$ for all $v \in S_G$.  Since
$\Encode_G$ is valid for $G$, receiver $e$ can compute $m_{w_G}$, and
since $\Encode_F$ is valid for $F$, this receiver can use the messages
indexed by  $\{w_G\} \times S_F$ along with $m_{w_G}$ to compute $(w_G,w_F)$.

Altogether, we have an encoding of $K$ bits using at most $(\beta(F)+\epsilon)(\beta(G)+\epsilon)K$ bits of the public channel, and the required result follows from letting $\varepsilon\to0$.
\end{proof}

\section{Linear programming}
\label{sec:lp}

In this section we derive a linear program whose value
constitutes a lower bound on the broadcast rate, and
we prove that the value of the LP behaves
supermultiplicatively under lexicographic products.
In fact, rather than working with a specific linear
program, we work with a general class of LP's
having two types of constraints: those dictated by the
network structure (which are the same for all LP's in
the general class), and
additional constraints depending only on
the vertex set, generated by a \emph{constraint schema},
i.e.\ a procedure for enumerating a finite set of constraints
given an arbitrary finite index set.
We identify some
axioms on the constraint schema that constitute a
sufficient condition for the LP value to be
supermultiplicative.
An example of a
constraint schema which is important in network information
theory is \emph{submodularity}.  For a given index set $I$,
the submodularity schema enumerates all of the
constraints of the form $z_S + z_T \geq z_{S \cap T} + z_{S \cup T}$
where $S,T$ range over subsets of $I$.

Now we explain the general class of LPs which behave
submultiplicatively under the lexicographic product and give bounds on
$\beta$.
Given an index code, if we sample each message independently
and uniformly at random, we obtain a finite probability space on
which the messages and the public channel are random
variables.  If $S$ is a subset of these random variables,
we will denote the Shannon entropy of the
joint distribution of the variables in $S$
by $H(S)$.  If $S \subseteq T \subseteq \clos(S)$ then every
message in $T \setminus S$ can be decoded
given the messages in $S$ and the public channel
$p$, and consequently $H(S \cup \{p\}) =
H(T \cup \{p\})$.  More generally, if we
normalize entropy (i.e.\ choose
the base of the logarithm) so that $H(x)=1$ for each
message $x$, then for every $S \subseteq T$ we have
\begin{equation} \label{eq:cst}
H(T \cup \{p\}) - H(S \cup \{p\}) \leq
|T \setminus \clos(S)|
\stackrel{\Delta}{=}
c_{ST},
\end{equation}
where the above is taken as the definition of $c_{ST}$.
This implies that for any index code we obtain
a feasible solution of the primal LP in Figure~\ref{fig:lp}
by setting $z_S = H(S \cup \{p\})$ for every $S$.
Indeed, the first constraint expresses the fact that the
value of $p$ is determined by the values of the
$n$ messages, which are mutually independent.
The second constraint was discussed above.
The final line of the LP represents a set of constraints,
corresponding to the rows of the matrix $\cm = (a_{qS})$,
that are universally valid for any tuple of random variables
indexed by the message set $\is$.  For
instance, it is well known that the entropy of
random variables has the submodularity property:
$H(S) + H(T) \ge H(S \cup T) + H(S \cap T)$
if $S,T$ are any two sets of random variables
on the same sample space.  So, for example, the
rows of
the constraint matrix $\cm$ could be indexed by
pairs of sets $S,T$, with entries in the $(S,T)$ row
chosen so that it represents the submodularity constraint
(namely $a_{q S} = a_{q T} = 1, \,
a_{q \, S \cap T} = a_{q \, S \cup T} = -1$ and all
other entries of row $a$ of $A$ are zero).
Noting that $H(\{p\}) \le \beta(G)$ we can altogether conclude the following theorem.

\begin{figure}[t]
\fbox{
\begin{minipage}[t]{0.3\textwidth}
{\small
\[ \begin{lparray}
\min & z_{\emptyset} &\\[\lplb]
\mbox{s.t.} & z_{\is} = |\is| & (w)\\[\lplb]
\forall S \subset T & z_T - z_S \leq c_{ST} & (x)\\[\lplb]
& Az \geq 0 & (y)
\end{lparray} \] }
\end{minipage}
\hfill
\qquad
\begin{minipage}[t]{0.6\textwidth}
{\small
\[ \begin{lparray}
\max & |\is| \cdot w - \sum_{S \subset T} c_{ST} x_{ST} \\[\lplb]
\mbox{s.t.} &
\sum_{q} a_{qS} y_q + \sum_{T \supset S} x_{ST} -
\sum_{T \subset S} x_{TS} = 0 &
\!\!\forall S \neq \emptyset, \is \\[\lplb]
& \sum_q a_{q \emptyset} y_q + \sum_{T \neq \emptyset} x_{\emptyset T} = 1
\\[\lplb]
& \sum_q a_{q \is} y_q - \sum_{T \neq \is} x_{T \is}  + w = 0
\\[\lplb]
& x,y \geq 0
\end{lparray} \]
}
\end{minipage}
}
\caption{The LP and its dual.}
\label{fig:lp}
\end{figure}

\begin{theorem}\label{thm:lplowerbound}
For an index coding problem $G$, let $\lpB(G)$ be the LP in
Figure \ref{fig:lp} when $A$ represents the submodularity constraints
and let $b(G)$ be its optimal solution.  Then $b(G) \le \beta(G)$.
\end{theorem}

It is known that entropies of sets of random variables
satisfy additional linear inequalities besides submodularity;
if desired, the procedure for constructing the matrix
$A$ could be modified to incorporate some of these
inequalities.
Alternatively, in the context of
restricted classes of encoding and decoding
functions (e.g.\ linear functions) there may be
additional inequalities that are specific to that
class of functions, in which case the constraint
matrix $A$ may incorporate these inequalities and
we obtain a linear program that is valid for this
restricted model of index coding but not valid in
general.  We will utilize such constraints in
Section~\ref{sec:nonlinear-linear-gap} when proving a separation
between linear and non-linear network coding.

\begin{definition} \label{def:schema}
A \emph{constraint schema}
associates to each finite index set $\is$
a finite set $\cset(\is)$ (indexing constraints)
and a matrix $\cm(\is)$ with rows indexed
by $\cset(\is)$ and columns indexed by
$\powerset{\is}$, the power set of $\is$.
In addition, to each Boolean
lattice homomorphism\footnote{A Boolean lattice
homomorphism preserves unions and intersections,
but does not necessarily map the empty set to
the empty set nor the universal set to the universal
set, and does not necessarily preserve complements.} $h : \powerset{\is}
\to \powerset{\otheris}$ it associates a function $h_* :
\cset(\is) \to \cset(\otheris)$.

Let $\ones$ be the $\powerset{\is}$-indexed vector
such that $\ones_S = 1$ for all $S$, and let
$\onesi$ be the vector where $(\onesi)_S = 1$ for
all $S$ containing $i$ and otherwise $(\onesi)_S=0$. We say that a constraint schema is
\emph{tight} if $A(I) \ones = A(I) \onesi = 0$
for every index set $\is$ and element $i \in \is$.

Given $h$ and $h_*$ let $P_h$ and $Q_h$ be matrices
representing the linear transformations
they induce on $\R^{\powerset{\is}} \to \R^{\powerset{\otheris}}$
 and $\R^{\cset(\is)} \to \R^{\cset(\otheris)}$,
respectively.  That is, $P_h$ and $Q_h$ have zeros
everywhere except $(P_h)_{h(S)S} = 1$ and
$(Q_h)_{h_*(q)q} = 1$.
We say that a constraint schema is \emph{homomorphic}
if it satisfies $A(\otheris)^{\trans} Q_h = P_h A(\is)^{\trans}$
for every Boolean lattice homomorphism
$h : \powerset{\is} \to \powerset{\otheris}$.
\end{definition}

\begin{example}
Earlier we alluded to the \emph{submodularity constraint schema}.
This is the constraint schema that associates to each index
set $\is$ the constraint-index set $\cset(\is) = \powerset{\is} \times
\powerset{\is}$, along with the constraint matrix
$\cm(\is)$ whose entries are as follows.  In row
$(S,T)$ and column $U$, we have an entry of $1$ if
$U = S$ or $U = T$, an entry of $-1$ if $U = S \cap T$ or $U = S \cup T$,
and otherwise 0.  (If any two of $S, \, T, \, S \cap T, \,
S \cup T$ are equal, then that row of $\cm(\is)$ is set to zero.)
It is easy to verify that $\cm(\is) \ones =
\cm(\is) \onesi = 0$ for all $i \in \is$, thus the schema is
tight.
For a homomorphism $h$, the corresponding
mapping of constraint sets is $h_*(S,T) = (h(S),h(T))$.
We claim that, equipped with this mapping of $h\to h_*$, the constraint schema is homomorphic. Indeed, to verify that $\cm(\otheris)^{\trans} Q_h = P_h \cm(\is)^{\trans}$ take any two sets $S, T \subset I$ and argue as follows to show that $u = P_h \, \cm(\is)^\trans \, \sbv_{S,T}$ and $v = \cm(\otheris)^\trans \, Q_h \, \sbv_{S,T}$ are identical (here and henceforth $\sbv_{X,Y}$ denotes the standard basis vector of $ \R^{\powerset{\is}}$ having $1$ in coordinate $(X,Y)$ for $X,Y\subset I$). First observe that
$\cm(\is)^\trans\, \sbv_{S,T}$ is the vector $\tilde{u} \in \R^{\powerset(\is)}$ which has $0$ entries everywhere except
$ \tilde{u}_{S} = \tilde{u}_{T} = 1$ and $\tilde{u}_{S\cup T} = \tilde{u}_{S\cap T} = -1$
provided that $S \nsubseteq T \nsubseteq S$, otherwise $\tilde{u}=0$.
As such, $u = P_h \tilde{u}$ has $0$ entries everywhere except
\begin{align*}
u_{h(S)} = u_{h(T)} = 1\,,\quad u_{h(S\cup T)} = u_{h(S\cap T)} = -1
\end{align*}
provided that $S \nsubseteq T \nsubseteq S$ and furthermore $h(S) \nsubseteq h(T) \nsubseteq h(S)$, otherwise $u=0$ (for instance, if $S \subseteq T$ then $\tilde{u} = 0$ and so $u=0$, whereas if $h(S) \subseteq h(T)$ then $\tilde{u}$ belongs to the kernel of $P_h$).
Similarly,
$Q_h \, \sbv_{S,T} = \sbv_{h(S),h(T)}$ and therefore $v = A(J)^\trans \, \sbv_{h(S),h(T)}$ has $0$ entries everywhere except
\begin{align*}
v_{h(S)} = v_{h(T)} = 1\,,\quad v_{h(S)\cup h(T)} = v_{h(S)\cap h(T)} = -1
\end{align*}
provided that $h(S) \nsubseteq h(T) \nsubseteq h(S)$, otherwise $v=0$.
To see that $u = v$ note that if $h(S) \subseteq h(T)$ then $u = v = 0$, and if $S \subseteq T$ then again we get $h(S) \subseteq h(T)$ due to monotonicity (recall that $h$ is a lattice homomorphism) and so $u=v=0$. Adding the analogous statements obtained from reversing the roles of $S,T$, it remains only to verify that $u=v$ in case $h(S) \nsubseteq h(T) \nsubseteq h(S)$, which reduces by the above definitions of $u$ and $v$ to requiring that $h(S \cup T) = h(S) \cup h(T)$ and $h(S \cap T) = h(S) \cap h(T)$. Both requirements are satisfied by definition of a Boolean lattice homomorphism, and altogether we conclude that the submodularity constraint schema is homomorphic.
\end{example}

\begin{theorem} \label{thm:lexprod}
Let $A$ be a tight homomorphic constraint schema.
For every index coding problem
let $\lpv(G)$ denote the optimum of the LP in Figure~\ref{fig:lp}
when $\is = V(G)$ and the constants $c_{ST}$ are
defined as in~\eqref{eq:cst}.
Then for every two index coding problems $G$ and $F$,
we have $\lpv(G \lexp H) \geq \lpv(G) \, \lpv(F)$.
\end{theorem}

\begin{proof}
It will be useful to rewrite the constraint set
of the dual LP in a more succinct form.
First, if $x$ is any vector indexed by pairs
$S,T$ such that $S \subset T \subseteq I$,
let $\nabla x \in \R^{\powerset{\is}}$ denote the
vector such that for all $S$,
$(\nabla x)_S = \sum_{T \supset S} x_{ST} -
\sum_{T \subset S} x_{TS}$.  Next,
for a
set $S \subseteq \is$, let $\sbv_S$ denote the
standard basis vector
vector in $\R^{\powerset{\is}}$ whose $S$ component
is $1$.
Then the entire constraint set of the dual
LP can be abbreviated to the following:
\begin{align}
\label{eq:dual-constraint}
A^\trans y + \nabla x + w \sbv_{\is} & = \sbv_{\emptyset}\,,
\quad x,y \ge 0\,.
\end{align}
Some further simplifications of the dual can be
obtained using the fact that the constraint schema
is tight.  For example, multiplying the left and right sides
of~\eqref{eq:dual-constraint} by the row vector
$\ones^\trans$ gives
\[
\ones^\trans A^\trans y + \ones^\trans \nabla x + w
= 1\,.
\]
By the tightness of the constraint schema $\ones^\trans A^\trans = 0$.
It is straightforward
to verify that $\ones^{\trans} \nabla x = 0$ and after eliminating
these two terms from the equation above, we find simply that $w=1$.
Similarly, if we multiply the left and right sides
of~\eqref{eq:dual-constraint} by the row vector
$\onesi^\trans$ and substitute $w=1$, we obtain
$
\onesi^\trans A^\trans y + \onesi^\trans \nabla x + 1  = 0
$
and consequently (again by the tightness) we arrive at $1  = - \onesi^\trans \nabla x$.
At the same time, $-\onesi^\trans\nabla x = \sum_{\substack{S \subset T \\ i \in T \setminus S}} x_{ST}$ by definition of $\nabla x$, hence summing over all $i \in I$ yields
 \[
 |\is| = \sum_{S \subset T} |T \setminus S| \, x_{ST}.
 \]
Plugging in this expression for $|\is|$ and
$w=1$, the LP objective
of the dual can be rewritten as
\[
|\is| - \sum_{S \subset T} c_{ST} x_{ST} =
\sum_{S \subset T} \left( |T \setminus S| - c_{ST} \right) \,  x_{ST} =
\sum_{S \subset T} |T \cap (\clos(S) \setminus S)| \, x_{ST},
\]
where the last equation used the fact that
$c_{ST} = |T \setminus \clos(S)|$.
We now define
\[\dualcoeff(S,T) = |T \cap (\clos(S)\setminus S)|\]
and altogether we arrive at the following reformulation of the dual LP.
\begin{equation} \label{eq:dual}
\begin{lparray}
\max & \sum_{S \subset T} \, \dualcoeff(S,T) \, x_{ST} \\[\lplb]
\mbox{s.t.} & A^\trans y + \nabla x = \sbv_\emptyset - \sbv_\is \\[\lplb]
& x,y \geq 0\,.
\end{lparray}
\end{equation}

Now suppose that $(\xg,\yg), (\xf,\yf)$ are optimal
solutions of the dual LP for $G,F$, achieving objective
values $\lpv(G)$ and $\lpv(F)$, respectively. (Here $\xi,\eta$ play the role of $x,y$ from~\eqref{eq:dual}, resp.) We will
show how to construct a pair of vectors $(\xgf,\ygf)$ that
is feasible for the dual LP of $G \lexp F$ and achieves an
objective value of at least $\lpv(G) \, \lpv(F)$.
The construction is as follows.
Let $g : \powerset{V(G)} \to \powerset{V(G \lexp F)}$
be the mapping $g(X) = X \times V(F)$.
For sets $S \subset T \subseteq V(G)$, let $h^{ST}: \powerset{V(F)}
\to \powerset{V(G \lexp F)}$ be the mapping $h^{ST}(X) = (T \times X)
\cup (S \times V(F)).$  Observe that both mappings are Boolean lattice homomorphisms.

To gain intuition about the mappings $g,h^{ST}$ it is useful to think
of obtaining the vertex set of $G \lexp F$ by replacing every
vertex of $G$ with a copy of $F$.  Here $g(\{v\})$
maps the vertex $v$ in $G$ to the copy of $F$ that replaces $v$.  The
mapping $h^{ST}(\{u\})$ maps a vertex $u$ in $F$ to the vertex $u$ in the
copies of $F$ that replace vertices in $T$, and then adds the set $\{u\}
\times V(F)$.

Recall that Definition~\ref{def:schema} associates
two matrices $P_h,Q_h$ to every Boolean lattice homomorphism
$h : \powerset{\is} \to \powerset{\otheris}$.
It is also useful to define a matrix $R_h$ as follows: the columns
and rows of $R_h$ are indexed by pairs $S \subset T \subseteq \is$
and $X \subset Y \subseteq \otheris$, respectively,
with the entry in row $XY$ and column $ST$ being equal
to 1 if $X = h(S)$ and $Y = h(T)$, otherwise 0.
Under this definition,
\begin{equation}\label{eq-nabla-Rh}
\nabla(R_h x) = P_h \nabla x \quad\mbox{ for any $x \in \R^{\powerset{\is}}$}\,.
\end{equation}
Indeed, if $x = \sbv_{S,T}$ for some $S\subset T \subseteq \is$ then
$\nabla \sbv_{S,T} = \sbv_S - \sbv_T$ and so $P_h\, \sbv_{S,T} = \sbv_{h(S)}-\sbv_{h(T)}$,
whereas $\nabla(R_h \sbv_{S,T}) = \nabla(\sbv_{h(S),h(T)}) = \sbv_{h(S)}-\sbv_{h(T)}$.

We may now define
\begin{align}
\xgf & =  \sum_{S \subset T} (\xg)_{ST} \, (R_{h^{ST}}\, \xf) \label{eq-xgf}\,,\\
\ygf & =  Q_g\, \yg + \sum_{S \subset T} (\xg)_{ST} \, (Q_{h^{ST}\,} \yf)\label{eq-ygf}\,.
\end{align}

In words, the dual solution for $G \lexp F$ contains a copy of the dual
solution for $F$ lifted according to $h^{ST}$ for every pair $S
\subset T$ and
one copy of the dual solution of $G$ lifted according to $g$.  The
feasibility of $(\xgf,\ygf)$ will follow from multiple applications of the
homomorphic property of the constraint schema and the feasibility of
$(\xf,\yf)$ and $(\xg,\yg)$, achieved by the following claim.

\begin{claim}\label{claim:feas}
The pair $(\xgf,\ygf)$ as defined in~\eqref{eq-xgf},\eqref{eq-ygf}
 is a feasible dual solution.
\end{claim}

\begin{proof}
The matrices $Q_g, \, R_{h^{ST}},  \, Q_{h^{ST}}$ all have $\{0,1\}$-valued
entries thus clearly $\xgf,\ygf \geq 0$.  Letting
$A = A(G \lexp F)$, we must prove that $A^\trans \ygf + \nabla \xgf
= \sbv_{\emptyset} - \sbv_{V(G \lexp F)}.$
Plugging in the values of $(\xgf,\ygf)$ we have
\begin{align}
A^\trans \ygf + \nabla \xgf
&=
A^\trans Q_g \yg + \sum_{S \subset T} (\xg)_{ST} \,
(A^\trans Q_{h^{ST}} \yf) +
\sum_{S \subset T} (\xg)_{ST} \, \nabla (R_{h^{ST}}\, \xf)\,,\nonumber\\
& =
P_g A(G)^\trans \yg + \sum_{S \subset T} (\xg)_{ST} \,
\left(P_{h^{ST}} A(F)^\trans \yf + \nabla (R_{h^{ST}}\, \xf)\right)\,.
\label{eq-dual-feasible-form1}
\end{align}
where the second equality applied the homomorphic property of the constraint schema.
To treat the summation in the last expression above, recall~\eqref{eq-nabla-Rh} which implies that
\begin{align}\label{eq-dual-feasible-form2}
P_{h^{ST}} A(F)^\trans \yf + \nabla (R_{h^{ST}}\, \xf) &=
P_{h^{ST}} A(F)^\trans \yf + P_{h^{ST}} \nabla \xf =
P_{h^{ST}} (\sbv_\emptyset - \sbv_{V(F)}) \,,
\end{align}
with the last equality due to the fact that $(\xf,\yf)$ achieves the
optimum of the dual LP for $F$. Recalling that $P_h \sbv_S = \sbv_{h(S)}$ for any $h$ and
combining it with the facts $h^{ST}(\emptyset) = S\times V(F)$
and $g(S) = S \times V(F)$ gives
$ P_{h^{ST}} \sbv_\emptyset = \sbv_{S \times V(F)} = P_g \sbv_S $.
Similarly, since $h^{ST}(V(F)) = T\times V(F)$ we have
$ P_{h^{ST}} \sbv_{V(F)} = \sbv_{T \times V(F)} = P_g \sbv_T $, and
plugging these identities in~\eqref{eq-dual-feasible-form2} combined with~\eqref{eq-dual-feasible-form1}
gives:
\[ A^\trans \ygf + \nabla \xgf
= P_g \bigg[ A(G)^{\trans} \yg + \sum_{S \subset T} (\xg)_{ST} \,
(\sbv_S - \sbv_T) \bigg]\,. \]
Collecting together all the terms
involving $\sbv_S$ for a given $S \in \powerset{\is}$, we find
that the coefficient of $\sbv_S$ is
$\sum_{T \supset S} (\xg)_{ST} - \sum_{T \subset S} (\xg)_{ST} = (\nabla \xg)_S$.  Hence,
\begin{align*}
A^\trans \ygf + \nabla \xgf
&=
P_g \left[ A(G)^\trans \yg + \nabla \xg \right]
=
P_g \left[ \sbv_\emptyset - \sbv_{V(G)} \right]
=
\sbv_\emptyset - \sbv_{V(G \lexp F)}\,,
\end{align*}
where the second equality was due to $(\xg,\yg)$ achieving the optimum of the dual LP for $G$.
\end{proof}

To finish the proof,
we must evaluate the dual LP objective and show that it
is at least $\lpv(G) \, \lpv(F)$, as the next claim establishes:

\begin{claim}\label{claim:dual-lower}
The LP objective for the dual solution given in Claim~\ref{claim:feas} has
value at least $\lpv(G) \, \lpv(F)$.
\end{claim}

\begin{proof}
To simplify the notation, throughout this proof we will use $K,L$ to
denote subsets of $V(G \lexp F)$ while referring to subsets of
$V(G)$ as $S,T$ and to subsets of $V(F)$ as $X,Y$.
We have
\begin{align}
\sum_{K \subset L}
\dualcoeff(K,L) (\xgf)_{KL}
&=
\sum_{K \subset L}
\dualcoeff(K,L) \sum_{S \subset T} (\xg)_{ST} \, (R_{h^{ST}} \,\xf)_{KL} \nonumber\\
&=
\sum_{S \subset T} (\xg)_{ST} \bigg(
\sum_{K \subset L} \dualcoeff(K,L) \, (R_{h^{ST}} \,\xf)_{KL}
\bigg) \nonumber\\
&=
\sum_{S \subset T} (\xg)_{ST} \bigg(
\sum_{X \subset Y} \dualcoeff\big(h^{ST}(X),h^{ST}(Y)\big) \, (\xf)_{XY}
\bigg)\,,\label{eq-dual-lower-bound}
\end{align}
where the last identity is by definition of $R_h$.

At this point we are interested in deriving a lower bound on
$\dualcoeff\big(h^{ST}(X),h^{ST}(Y)\big)$, to which end we first need to
analyze $\clos_{G \lexp F}(h^{ST}(X))$.
Recall that $E(G\lexp F)$ consists of
all hyperedges of the form $(w,K)$ with
$w=(w_G,w_F)$ and
$K = (W_G \times V(F)) \cup (\{w_G\}\times W_F)$
for some pair of edges $(w_G,W_G) \in E(G)$ and $(w_F,W_F)\in E(F)$.
%
We first claim that for any $S\subset T$ and $X\subset V(F)$,
\begin{equation}
  \label{eq-lex-closure-1}
  \clos_{G\lexp F} \left(h^{ST}(X) \right) \setminus h^{ST}(X)
\; \supseteq  \;
\Big( \left( \clos_G(S) \setminus S \right) \cap T \Big) \times
\Big( \clos_F(X) \setminus X \Big)
\,.
\end{equation}
To show this, let $L\subseteq V(G \lexp F)$ denote the set on the
right side of~\eqref{eq-lex-closure-1}.  Note
that $L$ contains no
ordered pairs whose first component is in $S$ or whose
second component is in $X$, and therefore $L$ is disjoint
from $h^{ST}(X) = (T \times X) \cup (S \times V(F))$.
Consequently, it suffices to show that
$\clos_{G \lexp F}\left(h^{ST}(X)\right) \supseteq L$.
Consider any $w = (w_G,w_F)$ belonging to $L$.
As $w_G \in \clos_G(S) \setminus S$, there must exist an edge
$(w_G,W_G) \in E(G)$ such that $W_G \subseteq S$.
Similarly, there must exist an edge $(w_F,W_F) \in E(F)$
such that $W_F \subseteq X$.  Recall from the definition
of $L$ that $\{w_G\} \subseteq T$.  Now letting
$K = (W_G \times V(F)) \cup (\{w_G\} \times W_F)$,
we find that
$K \subseteq (S \times V(F)) \cup (T \times X) = h^{ST}(X)$
and that $(w,K) \in E(G \lexp F)$, implying that
$w \in \clos_{G \lexp F} \left( h^{ST}(X) \right)$ as desired.

Let $\hstx = h^{ST}(X)$ and $\hsty = h^{ST}(Y)$,
and recall that $\dualcoeff(\hstx,\hsty)$ is defined as
$\big| \big( \clos_{G\lexp F}(\hstx) \setminus \hstx \big) \cap \hsty \big|$.
Using~\eqref{eq-lex-closure-1} and noting that
$\hsty \supseteq (T \times Y)$ we find that
\[
\left( \clos_{G \lexp F}(\hstx) \setminus \hstx \right) \cap \hsty \; \supseteq \;
\Big( \left( \clos_G(S) \setminus S \right) \cap T \Big)
\times
\Big( \left( \clos_F(X) \setminus X \right) \cap Y \Big)
\]
and hence
\[
\dualcoeff(\hstx,\hsty) \geq
\left| \left( \clos_G(S) \setminus S \right) \cap T \right| \, \cdot \,
\left| \left( \clos_F(X) \setminus X \right) \cap Y \right| =
\dualcoeff(S,T) \, \dualcoeff(X,Y) \, .
\]
Plugging this bound into~\eqref{eq-dual-lower-bound} we find that
\begin{align*}
\sum_{K \subset L}
\dualcoeff(K,L) (\xgf)_{KL} &\geq
\sum_{S \subset T} (\xg)_{ST} \sum_{X \subset Y} \dualcoeff(S,T) \dualcoeff(X,Y)
(\xf)_{XY} \\
&=
\bigg(\sum_{S \subset T} \dualcoeff(S,T) (\xg)_{ST} \bigg)
\bigg( \sum_{X \subset Y} \dualcoeff(X,Y) (\xf)_{XY} \bigg) =
\lpv(G) \, \lpv(F)\,,
\end{align*}
as required.
\end{proof}
Combining Claims~\ref{claim:feas} and~\ref{claim:dual-lower} concludes the proof of the Theorem~\ref{thm:lexprod}.
\end{proof}
\begin{remark*}
The two sides of~\eqref{eq-lex-closure-1} are in fact equal for any non-degenerate index coding instances $G$ and $F$, namely
under the assumption that every $(w_G,W_G)\in E(G)$ has
$w_G \notin W_G$ (otherwise this receiver already knows the required $w_G$
and may be disregarded) and $W_G \neq \emptyset$ (otherwise the public channel must include $w_G$ in plain form and we may disregard this message), and similarly for $F$.
To see this, by definition of $\clos_{G\lexp F}(\cdot)$ and the fact that
$h^{ST}(X) = (T \times X) \cup (S \times V(F))$ it suffices to show that every edge $(w,K)\in E(G\lexp F)$ with $K\subseteq h^{ST}(X)$
satisfies $w \in \big(\clos_G(S) \cap T\big)\times \clos_F(X)$. Take $(w,K) \in E(G\lexp F)$
and let $(w_G,W_G)\in E(G)$ and $(w_F,W_F)\in E(F)$ be the edges forming it as per Definition~\ref{def-lexicographic-prod} of the lexicographic product.
A prerequisite for $K \subseteq h^{ST}(X)$
is to have $w_G \in T$ as otherwise $\{w_G\}\times W_F \not\subseteq h^{ST}(X)$
(recall that $S \subset T$ and that $W_F \neq \emptyset$). Moreover, as $X$ is strictly contained in $V(F)$
we must have $W_G \subseteq S$ in order to allow $W_G \times V(F) \subseteq h^{ST}(X)$,
thus (using the fact that $w_G \notin W_G$ and so $w_G \notin S$) we further require
that $W_F \subseteq X$. Altogether we have $W_G \subseteq S$, $W_F \subseteq X$ and $w_G \in T$,
hence $(w_G,w_F) \in \big(\clos_G(S) \cap T\big)\times \clos_F(X)$ as required.
\end{remark*}

\section{Separation between $\alpha$ and $\beta$}
\label{sec:alpha-beta}
To prove Theorem~\ref{thm:alpha-beta-gap}, we start by using Theorem~\ref{thm:lplowerbound} to show that $\beta(C_5) > \alpha(C_5)$ where
$C_5$ is the 5-cycle.  Then we apply the power of Theorem~\ref{thm:lexprod} to transform this constant gap on $C_5$ to a
polynomial gap on $C_5^k$.

First we show that $\beta(C_5) \ge b(C_5) \ge \frac52$. We can show that $b(C_5) \ge \frac52$ by providing a feasible dual
solution for the LP $\lpB$ with value $\frac52$.  This can easily be achieved
by listing a set of primal
constraints whose variables sum and cancel to show that $z_\emptyset
\ge \frac52$.  Labeling the vertices of $C_5$ by $1,2,3,4,5$ sequentially,
such a set of constraints is given below.  It is helpful to note that
in an index coding problem defined by an undirected graph, $x \in \clos(S)$
if all the neighbors of $x$ are in $S$.
\begin{align*}
2 &\ge z_{\{1,3\}} - z_\emptyset  \\
2 &\ge z_{\{2,4\}} - z_\emptyset  \\
1 &\ge z_{\{5\}} - z_\emptyset  \\
0 & \ge z_{\{1,2,3\}} - z_{\{1,3\}} \\
0 & \ge z_{\{2,3,4\}} - z_{\{2,4\}} \\
z_{\{2,3,4\}} + z_{\{1,2,3\}} &\ge z_{\{2,3\}} + z_{\{1,2,3,4\}} \\
z_{\{2,3\}} + z_{\{5\}} &\ge z_\emptyset + z_{\{2,3,5\}} \\
0 &\ge z_{\{1,2,3,4,5\}} - z_{\{1,2,3,4\}} \\
0 &\ge z_{\{1,2,3,4,5\}} - z_{\{2,3,5\}} \\
z_{\{1,2,3,4,5\}} &= 5\\
z_{\{1,2,3,4,5\}} &= 5
\end{align*}
Applying Theorem~\ref{thm:lexprod} we deduce that for any integer
 $k\geq 1$ the $k$-th lexicographic power of $C_5$ satisfies
 $\beta(C_5^k) \ge b(C_5^k) \geq \left(\frac{5}{2}\right)^{k}$.
Furthermore, $\alpha(C_5) = 2$ and it is well known that the independence
number is multiplicative on lexicographic products and so $\alpha(C_5^k) = 2^k$.
Altogether, $C_5^k$ is a graph on $n = 5^k$ vertices with $\alpha = n^{\log_5(2)}$ and $\beta \ge
n^{1-\log_{5}(2)}$, implying our result.

\section{Matroids and index coding}
\label{sec:matroids}
Recall that a matroid is a pair $M = (E,r)$ where $E$ is a ground set and $r:2^E\to\N$ is
a rank function satisfying
\begin{compactenum}[(i)]
\item $r(A) \le |A|$ for all $A \subseteq E$;
\item $r(A) \le r(B)$ for all $A \subseteq B \subseteq E$ (monotonicity);
\item $r(A) + r(B) \ge r(A \cup B) + r(A \cap B)$ for all $A,B \subseteq E$ (submodularity).
\end{compactenum}
The rank vector of a matroid, $\rankvec(M)$, is a $2^{|E|}$-dimensional
vector indexed by
subsets of $E$, such that its $S$-th coordinate is $r(S)$.
A subset $S \subseteq E$ is called \emph{independent} if
$r(S) = |S|$ and it is called a \emph{basis} of $M$ if
$r(S) = |S| = r(E)$.

In this section we give a construction mapping a matroid to an
instance of index coding that exactly captures the
dependencies in the matroid.  We proceed to show some useful connections between
matroid properties and the broadcast rate of the corresponding
index coding problem.

\begin{definition}\label{def-matroid-index-coding}
Let $M = (E,r)$ be a matroid. The
hypergraph index coding problem \emph{associated} to $M$, denoted by $G_M$, has a message set
$E$ and all receivers of the form
\[\big\{(x,S) \,\mid\, x \in E\,,\, S \subseteq E\,,\,r(S) = r(S \cup \{x\})\big\}\,.\]
\end{definition}
\begin{remark*}
A similar yet slightly more complicated construction was given 
in~\cite{RSG2}. Our construction is (essentially) a subset of the one
appearing there.  A construction that maps a matroid to a network
coding problem is given in~\cites{DFZ1,DFZ2}.  They prove an analog of
Proposition \ref{prop:matroidLP}.
\end{remark*}


\begin{proposition}
For a matroid $M = (E,r)$, $b(G_M) = |E|-r(E)$.
\label{prop:matroidLP}
\end{proposition}

\begin{proof}
In what follows we will let $n = |E|$ and $r = r(E)$.
To show that $b(G_M) \le n-r$ it suffices to show $z_S = r(S) + n-r$ is a
feasible primal solution to the LP $\lpB(G_M)$.  The feasibility of
constraints $(w)$ and $(x)$ follows trivially from the the definition of $G_M$ and
properties of a matroid.  The feasibility of $(y): z_T - z_S \leq
c_{ST} \; \forall S \subset T$ follows from repeated application of submodularity:
\begin{align*}
z_T - z_S = r(T) - r(S) & \leq \sum_{x \in T \setminus S} r(S \cup \{x\})
- r(S) \\
& \leq \sum_{x \in \clos(S)} (r(S \cup \{x\}) - r(S)) \; + \;
 \sum_{x \in T \setminus \clos(S)} r(\{x\}) \leq |T \setminus \clos(S)| = c_{ST}.
\end{align*}
To prove the reverse inequality,
let $S$ be any basis of $M$ 
and note that
$
z_\emptyset = z_E - (z_E - z_S) - (z_S - z_\emptyset) \geq
n - c_{SE} - c_{\emptyset S} = n-r.
$
\end{proof}

The following definition relaxes the notion of a representation for a matroid.
\begin{definition}\label{def-under-rep}
A matroid $M = (E,r)$ with $|E| = n$ is \emph{under-representable} in
$d$ dimensions over a finite field $\F$ if there exists a $d \times n$
matrix with entries in $\F$ and columns indexed by elements of $E$
such that (i) the rows are independent and (ii) if $r(x \cup S) = r(S)$ then
the columns indexed by $x \cup S$ are dependent.
\end{definition}
Observe that if a matrix represents $M$ then it also under-represents $M$.
We next show a relation between under-representations for $M$ over $\F$ and the \emph{scalar} linear rate $\linrate^{\F}_1$,
where the alphabet vector space, over which the encoding functions are required to be linear, is single-dimensional.  Note that $\linrate^{\F} \leq \linrate^{\F}_1$.
The following is the analogue of Theorem~8 in~\cite{RSG2} for our version of the matroid to index coding mapping.
\begin{theorem}\label{thm:beta_rep_matroids}
A matroid $M = (E,r)$ with $|E| = n$ is under-representable in $d$ dimensions over
a finite field $\F$ if and only if $\linrate_1^{\F}(G_M) \le n-d$.
In particular, if $M$ is representable over $\F$ then
$\linrate^{\F}(G_M) = \beta(G_M) = n-r(E)$.
\end{theorem}
\begin{proof}
Let $R$ be a $d\times n$ matrix which under-represents $M$ in $d$ dimensions over $\F$.
Let $Q$ be an $(n-d)\times n$ matrix whose rows span the kernel of $R$.
We will show that $Q$ is a valid encoding matrix for $G_M$. Let $y \in \F^{E}$ be some input message set and consider a receiver $(x,S)$,
who wishes to decode $y_x$ from $\{y_z : z\in S\}$ and the broadcast message $Qy$.
Extend $\ker(Q)$ arbitrarily into a basis $B$ for $\F^{E}$ and let $y=y'+y''$ be the unique decomposition according to $B$
such that $y' \in \ker(Q)$. Clearly, $Q y'' = Q y$ since $y'\in \ker(Q)$, hence one can recover $y''$ from the public channel by triangulating $Q$. It remains for the receiver $(x,S)$ to recover $y'_x$. To this end, observe that the rows of $R$ span $\ker(Q)$ and
recall that by Definitions~\ref{def-matroid-index-coding} and~\ref{def-under-rep}, column $x$ of $R$ is a linear combination of the columns of $R$ indexed by $S$. Since $y'$ is in the row-space of $R$ it follows that $y'_x$ is equal to the exact same linear combination of the components of $y'$ indexed by $S$, all of which are known to the receiver. Altogether, the receiver can recover both $y'_x$ and $y''_x$ and obtain the message $x$. As this holds for any receiver, we conclude that $Q$ is a valid encoding matrix and thus $\linrate^{\F}_1(G_M)\leq n-d$.  When $d=r(E)$ the inequality is tight because this upper bound coincides with the lower bound given by Proposition~\ref{prop:matroidLP}.


Conversely, suppose that there exists a scalar linear code for $G_M$ over $\F$ with
rate $n-d$, and let $Q$ be a corresponding $(n - d) \times n$
encoding matrix of rank $n-d$.  Let $R$ be a $d
\times n$ matrix whose rows span the kernel of $Q$. We claim that $R$ under-represents $M$. Indeed, consider a
receiver $(x,S)$.  It is easy to verify that this receiver has a linear decoding function\footnote{This follows e.g.\ from decomposing $y$ as above into $y'+y''$ where $y' \in \ker(Q)$. By definition $y''_x$ is a linear combination of the $Q y$ entries. Similarly, $y'_x$ must be a linear combination
of $\{y_z:z\in S\}$, otherwise there would exist some $y\in\ker(Q)$ with $y_x \neq 0$ and $y_z = 0$ for all $z\in S$, making it indistinguishable to this receiver from $y=0$.} of the form $u^\trans \cdot
Q y + v^\trans \cdot y_S$ for some vectors $u,v$, where $y_S$ is the vector formed by restricting $y$ to the indices of $S$.
  As $Q$ is a valid encoding matrix for $G_M$, this evaluates to $y_x$ for any $y\in \F^E$.
In particular, if $y^\trans$ is a row of $R$ then $Q y  = 0$ and so $v^\trans \cdot y_S = y_x$, and applying this argument to every row of $R$ verifies
that column $x$ of $R$ is a linear combination of the columns of $R$ indexed by $S$ (with coefficients from $v$).  Since this holds for any receiver we have
that $R$ under-represents $M$, as required.
\end{proof}


We conclude this section with a result that will be useful in establishing
lower bounds on the value of the LP for $G_M$ with a given constraint matrix $A$.

\begin{theorem}
Suppose that $M=(E,r)$ is a matroid and $A$ is a matrix such that
$A \ones = 0$ and
$A \rankvec(M) \not\geq 0$.  If the linear program in Figure \ref{fig:lp} is
instantiated with constraint matrix $A$, then the value of the LP
is strictly greater than $|E| - r(E)$.
\label{thm:addIneq}
\end{theorem}

%
%
%

\newcommand{\splus}{\mathcal{S}^+}
\newcommand{\sminus}{\mathcal{S}^-}
\begin{proof}
We will give a dual solution $(w,x,y)$
to the LP with value strictly greater than $|E|-r(E)$.

Recalling the hypothesis $A \rankvec(M) \not\geq 0$, let
$q$ be a row of $A$ such that $\sum_{S\subseteq E} a_{qS} r(S) < 0$.
Let $\splus = \{S \subseteq E \mid a_{qS} > 0, \, S \ne E,\emptyset\}$ and
$\sminus = \{S \subseteq E \mid a_{qS} < 0, \, S \ne E,\emptyset\}$.   Note that the hypothesis
that $A \ones = 0$ implies that $a_{q\emptyset} + \sum_{S \in \splus}
a_{qS} = - \left(a_{qE} + \sum_{S \in \sminus} a_{qS} \right)$.
Assume that $A$ is scaled
so $a_{q\emptyset} + \sum_{S \in \splus}
a_{qS} = - \left(a_{qE} + \sum_{S \in \sminus} a_{qS} \right) = 1$.
This assumption is without loss of generality since
$a_{qE} + \sum_{S \in \sminus} a_{qS}$ is strictly
negative, as can be seen from the following calculation:
\begin{align*}
r(E)\bigg( a_{qE} + \sum_{S \in \sminus} a_{qS}\bigg) &\leq
a_{qE} r(E) + \sum_{S \in \sminus} a_{qS} r(S)
\leq
a_{qE} r(E) + \sum_{S \in \sminus} a_{qS} r(S) + \sum_{S \in \splus} a_{qS} r(S) \\
&=
\sum_{S} a_{qS} r(S) \;\; < \;\; 0\,.
\end{align*}

Define the dual vector $y$ by setting
$y_q=1$ and $y_{q'} =0$ for rows $q'\ne q$ of $A$.
To define the dual vector $x$, let us first
associate to every set $S \subseteq E$ a matroid basis
$b(S)$ such that the set $m(S) = b(S) \cap S$ is
a maximal independent subset of $S$, i.e.\ $|m(S)| = r(m(S)) = r(S)$.
Let $u(S) = S \cup b(S)$.
For every $S \in \splus$, let $x_{\emptyset m(S)} = x_{m(S) S} = a_{qS}$
and for every $S \in \sminus$, let
$x_{S u(S)} = x_{u(S) E} = -a_{qS}$.  Set all
other values of $x_{ST}$ to zero.  Finally, set $w= 1$.  By construction,
$(w,x,y)$ satisfies all of the dual constraints.
Using the relations $c_{\emptyset m(S)} = r(S), \,
c_{S u(S)} = r(E) - r(S), \, c_{m(S) S} = c_{u(S)E} = 0$,
we find that the dual LP objective value is
\begin{align*}
|E| \, w - \sum_{S \subset T} c_{ST} x_{ST}
&=
|E| - \sum_{S \in \splus} (c_{\emptyset m(S)} + c_{m(S) S}) a_{qS}
- \sum_{S \in \sminus} (c_{S u(S)} + c_{u(S) E})(- a_{qS}) \\
&=
|E| - \sum_{S \in \splus} r(S) a_{qS} +
\sum_{S \in \sminus} (r(E) - r(S)) a_{qS} \\
& = |E| + \sum_{S \in \sminus}a_{qS}r(E) - \sum_S a_{qS} r(S) +
a_{q\emptyset}r(\emptyset) + a_{qE}r(E)\\
& = |E| - r(E) - \sum_S a_{qS} r(S).
\end{align*}
By hypothesis $\sum_S a_{qS} r(S) < 0$, and the
proposition follows.
\end{proof}

\section{Separation between linear and non-linear rates}
\label{sec:nonlinear-linear-gap}

In this section we prove Theorem~\ref{thm:NWCgap}.  To this end we will
first show that the linear rate over a field of even characteristic is
strictly better than the linear rate over a field of odd
characteristic for the index coding problem associated to the Fano
matroid, and that the reverse relation holds
for the non-Fano matroid.  Then we will take
the lexicographic product of the two index codes to get a gap between
the linear and non-linear coding rates, and then use
lexicographic products again to amplify that gap.

The \emph{Fano matroid}, denoted $\fano$, and the \emph{non-Fano
  matroid}, denoted $\nonfano$, are 7 element, rank 3 matroids.  The seven columns of the matrix
\[
\begin{pmatrix}
1 & 0 & 0 & 0 & 1 & 1 & 1 \\
0 & 1 & 0 & 1 & 0 & 1 & 1 \\
0 & 0 & 1 & 1 & 1 & 0 & 1
\end{pmatrix}
\]
constitute a linear representation of the Fano matroid when $\chr(\F) =
2$ and one for the non-Fano matroid when  $\chr(\F) \neq 2$.  We will use $\;\groundset =
\{100,010,001,110,101,011,111\}$ to index the elements of the two
matroids.  Further let $\odds \subset \groundset$ be the vectors with odd
hamming weight, let $\oddsm$ be the vectors with hamming weight one and let $i+j$ for $i, j \in \groundset$ be the bitwise
addition of $i,j$.

It is well known that the Fano matroid is representable only in a
field of characteristic 2, and the non-Fano matroid is representable
in any field whose characteristic is different from 2 but not
in fields of characteristic 2.  We use
a generalization of this fact
to prove the following theorem that directly implies Theorem~\ref{thm:NWCgap}.

\begin{theorem}[Separation Theorem]
Let $G = G_\fano \lexp G_\nonfano$.
There exists some $\epsilon>0$ such that $\beta(G^{\lexp n}) = 16^n$
whereas $\linrate(G^{\lexp n}) \ge (16 + \epsilon)^n$ for all $n$.
\label{thm:separation}
\end{theorem}
The fact that $\beta(G^{\lexp n}) = 16^n$ will be a straightforward application of
Proposition~\ref{prop:matroidLP} and Theorem~\ref{thm:beta_rep_matroids}.
The lower bound on the linear rate however will require considerably more effort. In
order to bound $\linrate$ from below we will extend the LP $\lpB$ to
two LPs, one of which will be a lower bound for linear codes over
fields with odd characteristic and the other for linear codes over
even characteristic.  Each one will supplement the matrix $A$ in the
LP with a set of constraints, one set derived from dimension
inequalities based on the representation
of the the Fano matroid
and the other from the non-Fano matroid.  The LP
that gives a lower bound for linear codes over a field with even characteristic will be used to
show that the linear broadcast rate of $G_\nonfano$ over a field of
even characteristic is strictly greater than four, and the LP for odd
characteristic will imply the corresponding result for $G_\fano$.
Furthermore, the constraints will satisfy the conditions of Theorem~\ref{thm:lexprod}.  Putting this all together implies that when
we take the lexicographic product of the Fano and non-Fano index
coding problems, no linear code
is as good as one that combines linear codes over $\F_2$ and $\F_3$.

Before explaining how we derive these constraints, we introduce a bit
of notation.  If $\{V_i\}_{i \in I}$ are subspaces of a
vector space $V$, let the span of $V_i$ and $V_j$ be denoted $V_i+V_j$
and let $\ds(\{V_i\}_{i \in I})$ be the dimension
of the span of $\{V_i\}_{i \in I}$.  Also, let $\dimvec(\{V_i\}_{i \in I})$ be a
$2^{|I|}$ dimensional vector indexed by the subsets of $I$ such that
the coordinate indexed by $S$ is $\ds(\{V_i\}_{i \in S})$.
We let $V_1 \oplus \cdots \oplus V_k$ denote the sum of
mutually complementary subspaces $V_1,\ldots,V_k$.
If $V = V_1 \oplus \cdots \oplus V_k$ then $V$
is isomorphic to the vector space $\prod_{i=1}^k V_i$
via the mapping $(v_1,\ldots,v_k) \mapsto v_1 + \cdots + v_k$.
In this case, for an index set $S \subseteq \{1,\ldots,k\}$, we will
use $\pi_S$ to denote the projection function
$V \rightarrow \oplus_{i \in S} V_i$, i.e.\ the function
that maps an element $v = \sum_{i=1}^k v_i$ to the
element $\pi_S(v) = \sum_{i \in S} v_i.$

The fact that the Fano matroid can be represented over $\F_2$ and the
non-Fano matroid cannot tells us something about dimension dependencies
that can occur in $\F_2$.  The following lemma is extracting the critical dimension
relations that distinguish vector spaces over $\F$ with $\chr(\F) =2$.

\begin{lemma}
\label{lem:char2}
Let $V = V_1 \oplus V_2 \oplus V_3$ be a vector space
over a field $\F$, and suppose $W \subset V$ is a linear subspace
that is complementary to each of $V_1 \oplus V_2, V_1 \oplus V_3,
V_2 \oplus V_3$.  Then
\begin{equation} \label{eq:3proj}
\ds \left( \pi_{12}(W) , \pi_{13}(W) , \pi_{23}(W) \right) =
\begin{cases}
2 \ds(W) & \mbox{if $\chr(\F)= 2$} \\
3 \ds(W) & \mbox{if $\chr(\F)\neq 2$}.
\end{cases}
\end{equation}
\end{lemma}
\begin{proof}
Recalling that $V$ is isomorphic to $ \prod_{i=1}^3 V_i$, we will
write elements of $V$ as ordered triples.
Our assumption that $W$ is complementary to each
of $V_1 \oplus V_2, V_1 \oplus V_3, V_2 \oplus V_3$
implies that a nonzero element of $W$ has three nonzero
coordinates, a fact that we will use in both cases of the lemma.

If $\chr(\F) = 2$, then every vector $(x,y,z) \in V$ satisfies
\[
\pi_{12}(x,y,z) + \pi_{13}(x,y,z) = (x,y,0) + (x,0,z) = (0,y,z) = \pi_{23}(x,y,z)
\]
hence $\pi_{12}(W) + \pi_{13}(W) = \pi_{23}(W)$.  Consequently
\[
\ds \left(  \pi_{12}(W) ,\pi_{13}(W) , \pi_{23}(W) \right) =
\ds \left( \pi_{12}(W) , \pi_{13}(W) \right) \leq 2 \ds(W).
\]
To prove the reverse inequality we observe that
$\pi_{12}(W)$ and $\pi_{13}(W)$ are complementary,
since every nonzero element of
$\pi_{12}(W)$ is of the form $(x,y,0)$ with $x,y \neq 0$,
whereas every nonzero element of $\pi_{13}(W)$ is of the form
$(x,0,z)$ with $x,z \neq 0$, and hence $\pi_{12}(W) \cap \pi_{13}(W) = \{0\}$.

When $\chr(\F) \neq 2$, we prove Equation~\eqref{eq:3proj} by
showing that $\pi_{12}(W), \pi_{13}(W), \pi_{23}(W)$
are mutually complementary.  Consider any three vectors
$w_1=(x_1,y_1,z_1)$, $w_2=(x_2,y_2,z_2)$, and $w_3=(x_3,y_3,z_3)$,
all belonging to $W$, such that
\[0 = \pi_{23}(x_1,y_1,z_1) + \pi_{13}(x_2,y_2,z_2) + \pi_{12}(x_3,y_3,z_3)
= (x_2+x_3,y_1+y_3,z_1+z_2)\,.\]
This implies that $x_2+x_3=0$, so the first coordinate of
$w_2+w_3$ is zero.  However, the zero vector is the
only vector in $W$ whose
first coordinate is zero, hence $w_2+w_3=0$.
Similarly, $w_1+w_3=0$ and $w_1+w_2=0$.  Now using
the fact that $2$ is invertible in $\F$, we deduce
that
$w_1 = \frac12[(w_1+w_2) + (w_1+w_3) - (w_2+w_3)] = 0$, and similarly $w_2=0$ and $w_3=0$.
Thus, the only way to express the zero vector as a sum of vectors
in $\pi_{12}(W), \pi_{13}(W), \pi_{23}(W)$ is if all three summands
are zero, i.e.\ those three subspaces are mutually complementary as claimed.
\end{proof}

\subsection{Linear codes over fields of characteristic two}\label{sec:linear-even}
This section provides the ingredients for
proving that $\linrate^\F(G_\fano) > 4$ for $\F$ with $\chr(\F) = 2$.

\begin{lemma}[Conditional Even Characteristic Inequality]
Suppose $\{V_i\}_{ i \in \groundset}$ are 7 subspaces of a vector space
over $\F$ such that $\chr(\F) = 2$ and
\begin{enumerate}[(i)]
\item $\ds(\{V_i\}_{i \in \odds}) = \ds(\{V_i\}_{i\in \oddsm})$
\item $\ds(V_i, V_j,V_k) = \ds(V_i) + \ds(V_j) + \ds(V_k) \; \forall i,j,k \in \odds$
\item $\ds(V_i,V_j,V_{i+j}) = \ds(V_i, V_j) \; \forall i,j \in
  \odds$ \label{hyppairs}
\end{enumerate}
Then $\ds(V_{110},V_{101},V_{011}) \le 2\ds(V_{111})$.
\label{lem:conditional_fano_inequality}
\end{lemma}

\begin{proof}[Proof of Lemma~\ref{lem:conditional_fano_inequality}]
Hypotheses (i) and (iii) of the lemma imply that all 7 subspaces are
contained in the span of $V_{100},V_{010},V_{001}$.
Moreover, hypothesis (ii) implies
that $V_{100},V_{010},V_{001}$ are mutually
complementary and that $V_{111}$ is complementary
to each of $V_{100}+V_{010}, \,
V_{100}+V_{001}, \, V_{010}+V_{001}$.
Thus, we can apply
Lemma~\ref{lem:char2} with $V = V_{100}\oplus V_{010} \oplus V_{001}$
and $W=V_{111}$,  yielding the equation
$ \ds(\pi_{12}(V_{111}),\pi_{23}(V_{111}), \pi_{13}(V_{111}))
= 2\ds(V_{111}).
$

We claim that $\pi_{12}(V_{111}) = (V_{001}+V_{111}) \cap (V_{100}+V_{010})$.
To see this, take an arbitrary element $w \in V_{111}$ having a unique
representation of the form
$x+y+z$ with $x \in V_{100}, y \in V_{010}, z \in V_{001}$.
By definition $\pi_{12}(w) = x+y = w-z,$ from which it can
be seen at once that $\pi_{12}(w)$ belongs to both
$V_{100} + V_{010}$ and $V_{001}+V_{111}$.  Conversely,
any element $v \in (V_{001}+V_{111}) \cap (V_{100}+V_{010})$
can be expressed as $v=w-z$ where $w \in V_{111}, z \in V_{001}$
but it can also be expressed as $v=x+y$ where
$x \in V_{100}, y \in V_{010}$.  Consequently,
$w = x+y+z$ and $v = \pi_{12}(w)$.

Hypothesis (iii) implies that $V_{110}$ is contained in both
$V_{001}+V_{111}$ and $V_{100}+V_{010}$, hence
$V_{110} \subseteq \pi_{12}(V_{111})$.  Similarly
$V_{101} \subseteq \pi_{13}(V_{111})$ and
$V_{011} \subseteq \pi_{23}(V_{111})$.
Hence
$\ds(V_{110},V_{101},V_{011}) \le \ds(\pi_{12}(V_{111}),\pi_{23}(V_{111}), \pi_{13}(V_{111}))
= 2\ds(V_{111})$, as desired.
\end{proof}

In what follows we will transform the conditional inequalities given
in the lemma above to a general inequality that
applies to any 7 subspaces of a vector space over a field of characteristic $2$
by using the following approach. We will start with arbitrary
subspaces and then repeatedly modify them until they satisfy the conditions of
Lemma~\ref{lem:conditional_fano_inequality}. At that point
the result in this conditional lemma will imply an inequality
involving the dimensions of the modified subspaces, which we will
express in terms of the dimensions of the original subspaces.

\begin{theorem}[Even Characteristic Inequality]\label{thm:fano_inequality}
There exists a $2^{7}$-dimensional vector
$\Lambda_{\mathrm{even}}$ such that for any 7 subspaces $\{V_i\}_{ i \in \groundset}$ of a vector space
over $\F$ with $\chr(\F) = 2$,
$$ \Lambda_{\mathrm{even}} \cdot \dimvec(\{V_i\}_{i \in \groundset}) \ge 0
\text{ and }\Lambda_{\mathrm{even}} \cdot \rankvec(\nonfano)< 0.$$
\end{theorem}
\begin{proof}
As mentioned above, the proof will proceed by repeatedly modifying the input subspaces until they satisfy the requirements of Lemma~\ref{lem:conditional_fano_inequality}.
The modifications we make to a vector space are of one type: we
\emph{delete} a vector $w$ from a subspace $V$ that contains $w$, by
letting $B$ be a basis of $V$ containing $w$ and then
replacing $V$ with the span of $B \setminus w$.

Let $\{V_i\}_{i \in \groundset}$ be seven subspaces of a vector space
$V$ over $\F$ such that $\chr(\F) = 2$.  We will modify the subspaces
$\{V_i\}_{ i \in \groundset}$ into $\{V'_i\}_{i \in \groundset}$ that satisfy the conditions of
Lemma~\ref{lem:conditional_fano_inequality}.  To start, we set $\{V'_i\}_{i \in
  \groundset} = \{V_i\}_{i \in \groundset}$. We then update
$\{V'_i\}_{i \in \groundset}$ in three steps, each of which
 deletes vectors of a certain type in an iterative fashion.  The order
 of the deletions within each step is arbitrary.

\begin{enumerate}[Step 1:]
\item Vectors in $V'_{111}$ but not in 
  $\sum_{i \in \oddsm} V'_i$ from $V'_{111}$.

\item
\begin{enumerate}[(a)]
\item Vectors in $V'_{100} \cap V'_{010}$ from $V'_{010}$.
\item Vectors in $V'_{001} \cap (V'_{100} + V'_{010})$ from $V'_{001}$.
\item Vectors in $V'_{111} \cap (V'_{100} + V'_{010})$ from $V'_{111}$.
\item Vectors in $V'_{111} \cap (V'_{010} + V'_{001})$ from $V'_{111}$.
\item Vectors in $V'_{111} \cap (V'_{100} + V'_{001})$ from $V'_{111}$.
\end{enumerate}

\item  Vectors in $V'_{i+j}$ but not in
$V'_i + V'_j$ for $i,j \in \odds$ from $V'_{i+j}$.
\end{enumerate}

First, we argue that $\{V'_i\}_{i \in \groundset}$ satisfy the
conditions of Lemma~\ref{lem:conditional_fano_inequality}.
The deletions in step (1) ensure that $V'_{111}$ is contained
in $\sum_{i \in \oddsm} V'_i$, thus satisfying condition (i).
The deletions in steps \mbox{(2a)--(2b)} ensure that $V'_{100},V'_{010},V'_{001}$
are mutually complementary, and steps \mbox{(2c)--(2d)} ensure that
$V'_{111}$ is complementary to the sum of any two of them,
thus satisfying condition (ii).
Furthermore, step
(2) does not change $\sum_{i \in \oddsm} V_i'$ because we only delete
a vector from one of $\{V'_i\}_{i \in \oddsm}$ when it belongs to
the span of the other two.  Thus condition (i) is
still satisfied at the end of step (2).  Step (3) ensures
that $V'_{i+j}$ is contained in $V'_i + V'_j$, thus satisfying
condition (iii).  Furthermore, it does
not modify $V'_i, i \in \odds,$ and thus conditions (i) and (ii)
remain satisfied after step (3).

Now, by Lemma~\ref{lem:conditional_fano_inequality} we have that
\begin{equation}
\ds(V'_{110},V'_{101},V'_{011}) \le 2\ds(V'_{111}).
\label{eq:uncondV'2}
\end{equation}

Let
\begin{align*}
\delta & = \ds(V_{111}, \{V_i\}_{i \in \oddsm}) - \ds(\{V_i\}_{i \in \oddsm}) \\
\delta[i|j,k] & = \ds(V_i,V_j,V_k) - \ds(V_j,V_k) \\
\delta[i;j] & = \ds(V_i \cap V_j) = \ds(V_i) + \ds(V_j) - \ds(V_i,V_j)  \\
\delta[i;j,k] &= \ds(V_i \cap (V_j + V_k) ) =
\ds(V_i) + \ds(V_j,V_k) - \ds(V_i,V_j,V_k)
\end{align*}
Observe that after step (1) $\ds(V'_{111}) = \ds(V_{111}) - \delta$, and steps (2) and (3) only delete more vectors from $V'_{111}$, so we have $\ds(V'_{111}) \le \ds(V_{111}) - \delta$.

It remains to get a lower bound on  $\ds(V'_{110},V'_{101},V'_{011})$ in terms of dimensions of subsets of $\{V_i\}_{i \in \groundset}$.
We do this by giving an upper bound on the total number of vectors deleted from $E = V'_{110}+V'_{101}+V'_{011}$ in terms
of the $\delta$ terms we defined above.  In steps (1) and (2) we delete nothing from $E$, but we delete some vectors from $V'_i, i \in \odds$.  Specifically,
$\delta[100;010]$ vectors are deleted from $V'_{010}$,
$\delta[001;100,010]$ vectors are deleted from $V'_{001}$,
and no vectors are deleted from $V_{100}$.
As already noted, step (1) deletes $\delta$ vectors from $V'_{111}$,
while step (2) deletes at most $\sum_{i,j \in \oddsm} \delta[111;i,j]$
vectors from $V'_{111}$.  To summarize, the dimensions of
$V'_i, i \in \odds,$ after steps (1) and (2), satisfy:
\begin{align}
\label{eq:fano-100}
\ds(V'_{100}) &= \ds(V_{100}) \\
\label{eq:fano-010}
\ds(V'_{010}) &= \ds(V_{010}) - \delta[100;010] \\
\label{eq:fano-001}
\ds(V'_{001}) &= \ds(V_{001}) - \delta[001;100,010] \\
\label{eq:fano-111}
\ds(V'_{111}) & \geq \ds(V_{111}) - \delta -
\sum_{i,j \in \oddsm} \delta[111;i,j].
\end{align}

In step (3), when we delete vectors in $V'_{i+j}$ but not in $V'_i + V'_j$,
if no deletions had taken place in prior steps then the number of vectors
deleted from $V'_{i+j}$ would be $\delta[i+j|i,j]$.  However, the
deletions that took place in steps (1) and (2) have the effect of
reducing the dimension of $V'_i + V'_j$, and we must adjust our
upper bound on the number of vectors deleted from $V'_{i+j}$ to
account for the potential difference in dimension between
$V_i+V_j$ and $V'_i+V'_j$.  When $i=100, \, j=010$, there
is no difference between $V_i + V_j$ and $V'_i + V'_j$,
because the only time vectors are deleted from either one
of these subspaces is in step (2a), when vectors in
$V'_{100} \cap V'_{010}$ are deleted from $V'_{010}$
without changing the dimension of $V'_{100} + V'_{010}$.
For all other pairs $i,j \in \odds$, we
use the upper bound
\[
\ds(V_i + V_j) - \ds(V'_i + V'_j) \leq
\left[ \ds(V_i) - \ds(V'_i) \right] +
\left[ \ds(V_j) - \ds(V'_j) \right],
\]
which is valid for any four subspaces $V_i,V_j,V'_i,V'_j$
satisfying $V'_i \subseteq V_i, \, V'_j \subseteq V_j$.
Let $\dds(V_i)$ denote the difference
$\ds(V_i) - \ds(V'_i)$.
Combining these upper bounds, we
find that the number of extra vectors deleted from $E$
in step (3) because of differences in dimension between
$V'_i+V'_j$ and $V_i+V_j$
is at most
\begin{align*}
\Bigg( \sum_{i,j \in \odds}
\dds(V_i) + \dds(V_j)
\Bigg)
& - \dds(V_{100}) - \dds(V_{010})
\\
 & =
2 \Bigg( \sum_{i \in \{100,010\}}
\dds(V_i) \Bigg)
+
3 \Bigg( \sum_{i \in \{001,111\}}
\dds(V_i) \Bigg)
\\
 & \leq
2 \delta[100;010] +
3 \delta[001;100,010] +
3 \delta +
3 \sum_{i,j \in \oddsm} \delta[111;i,j]
\end{align*}
where the last inequality follows by combining equations
\eqref{eq:fano-100}--\eqref{eq:fano-111}.

We now sum up our upper bounds on the number of vectors
deleted from $E$ in step (3), to find
that
\begin{equation} \label{eq:fano-e}
\ds(E) \geq
\ds(V_{110},V_{101},V_{011}) -
\sum_{i,j \in \odds} \delta[i+j|i,j] -
2 \delta[100;010] -
3 \delta[001;100,010] -
3 \delta -
3 \sum_{i,j \in \oddsm} \delta[111;i,j].
\end{equation}
Expanding out all the $\delta$ terms,
combining with the upper bound $\ds(V'_{111}) \leq \ds(V_{111}) - \delta$,
and plugging these into
Equation \eqref{eq:uncondV'2} gives us
$\Lambda_{\mathrm{even}} \cdot \dimvec(\{V_i\}_{i
  \in \groundset}) \ge 0$ for
some $2^7$-dimensional vector $\Lambda_{\mathrm{even}}$, as desired;
after applying these steps one obtains
Equation~\eqref{eq:fano-full} below.
When $\{V_i\}_{i \in \groundset}$ are one-dimensional subspaces
constituting a representation of the non-Fano matroid
over a field of characteristic $\neq 2$,
it is easy to check that all of the $\delta$ terms
appearing in~\eqref{eq:fano-e} are zero.  So, the
inequality states that $\ds(V_{110},V_{101},V_{011}) \le
2 \ds(V_{111}),$ whereas we know that $\ds(V_{110},V_{101},V_{011}) =
3 \ds(V_{111})$ for the non-Fano matroid.  Consequently
$\Lambda_{\mathrm{odd}}\nolinebreak\cdot\rankvec(\fano)<0$.

For completeness, the inequality $\Lambda_{\mathrm{even}} \cdot
\dimvec(\{V_i\}_{i \in \groundset}) \ge 0$ is written
explicitly as follows.
\begin{align} \notag
& 2 \ds(V_{100})
+ 2 \ds(V_{010})
+ 3 \ds(V_{001})
+ 11 \ds(V_{111}) \\
+ & \notag
  3 \ds(V_{100},V_{010})
+ 2 \ds(V_{100},V_{001})
+ 2 \ds(V_{010},V_{001}) \\
- & \notag
    \ds(V_{100},V_{111})
-   \ds(V_{010},V_{111})
-   \ds(V_{001},V_{111})
- 4 \ds(V_{100},V_{010},V_{001})  \\
-  & \notag
  3 \ds(V_{111},V_{100},V_{010})
- 3 \ds(V_{111},V_{100},V_{001})
- 3 \ds(V_{111},V_{010},V_{001}) \\
+ & \notag
    \ds(V_{110},V_{100},V_{010})
+   \ds(V_{101},V_{100},V_{001})
+   \ds(V_{011},V_{010},V_{001}) \\
+ & \notag
    \ds(V_{110},V_{111},V_{001})
+   \ds(V_{101},V_{111},V_{010})
+   \ds(V_{011},V_{111},V_{100}) \\
- & \label{eq:fano-full}
    \ds(V_{110},V_{101},V_{011})
+   \ds(V_{111},V_{100},V_{010},V_{001}) \geq 0\,.
\end{align}
This concludes the proof of the theorem.
\end{proof}

\subsection{Linear codes over fields of odd characteristic}\label{sec:linear-odd}


The following lemma and theorem, which are analogues of Lemma~\ref{lem:conditional_fano_inequality} and Theorem~\ref{thm:fano_inequality} from Section~\ref{sec:linear-even},
  provide an inequality for fields with odd characteristic.

\begin{lemma}[Conditional Odd Characteristic Inequality]
Suppose $\{V_i\}_{ i \in \groundset}$ are 7 subspaces of a vector space
over $\F$ such that $\chr(\F) \ne 2$ and
\begin{enumerate}[(i)]
\item $\ds(\{V_i\}_{i \in \odds}) = \ds(\{V_i\}_{i\in \oddsm})$
\item $\ds(V_i, V_j,V_k) = \ds(V_i) + \ds(V_j) + \ds(V_k) \; \forall i,j,k \in \odds$
\item $\ds(V_i,V_j,V_{i+j}) = \ds(V_i, V_j) \; \forall i,j \in \oddsm$
\item $\ds(V_i,V_j,V_{111}) = \ds(V_i, V_j) \; \forall i,j: i+j = 111$
\end{enumerate}
Then $\ds(V_{110},V_{101},V_{011}) \ge 3\ds(V_{111})$.
\label{lem:conditional_nonfano_inequality}
\end{lemma}

\begin{proof}
Just as in the proof of Lemma~\ref{lem:conditional_fano_inequality} we
apply the result of Lemma~\ref{lem:char2}, but now with $\chr(\F) \ne 2$.
Hypotheses (i) and (iii) imply that all 7 subspaces are contained in the
span of $V_{100}, V_{010}, V_{001}$, and hypothesis (ii) implies that
those three subspaces are mutually complementary, and that
$V_{111}$ is complementary to the sum of any two of them.
Thus, Lemma~\ref{lem:char2} implies that $\ds(V_{110},V_{101},V_{011})
= 3 \ds(W)$.
Now
we aim to show that hypotheses (iii) and (iv) imply that $V_{110}$ contains
$\pi_{12}(V_{111})$, and similarly for $V_{101}, V_{011}$.  This will imply
that $\ds(V_{110},V_{101},V_{011}) \ge \ds(\pi_{12}(W),\pi_{23}(W), \pi_{13}(W))
= 3\ds(W)$ as desired.

It remains for us to
justify the claim that $V_{110}$ contains $\pi_{12}(V_{111})$.
Suppose $(x,y,z)$ belongs to $V_{111}$, where we use $(x,y,z)$ as an
alternate notation for $x+y+z$ such that $x$ belongs to $V_{100}$, $y$
belongs to $V_{010}$, $z$ belongs to $V_{001}$.  We know from
hypothesis (iv) that $V_{111}$ is contained in $V_{001}+V_{110}$.  So
write $x+y+z = a+b$ where $a$ is in $V_{001}$ and $b$ is in
$V_{110}$.  We know from hypothesis (iii) that $V_{110}$ is contained
in $V_{100}+V_{010}$, so write $b = c+d$ where $c$ is in $V_{100}$ and $d$ is in
$V_{010}$.  Then $x+y+z = c+d+a$, and both sides are a sum of three
vectors, the first belonging to $V_{100}$, the second to $V_{010}$,
the third to $V_{001}$.  Since those three vector spaces are mutually
complementary, the representation of another vector as a sum of
vectors from each of them is unique.  So $x=c, y=d, z=a$.  This means
that  $x+y=c+d=\pi_{12}(x,y,z)$.  Recall that $c+d$ is in $V_{110}$.  As
$(x,y,z)$ was an arbitrary element of $V_{111}$, we have shown
that  $V_{110}$ is contained in $\pi_{12}(V_{111})$.
\end{proof}

 \begin{theorem}[Odd Characteristic Inequality]\label{thm:nonfano_inequality}
There exists a $2^{7}$-dimensional vector $\Lambda_{\mathrm{odd}}$ such that for any
7 subspaces  $\{V_i\}_{ i \in \groundset}$ of a vector space
over $\F$ with $\chr(\F) \ne 2$,
\[ \Lambda_{\mathrm{odd}} \cdot \dimvec(\{V_i\}_{i \in \groundset}) \ge 0
\text{ and }\Lambda_{\mathrm{odd}} \cdot \rankvec(\fano) < 0\,.\]
\end{theorem}

\begin{proof}
Let $\{V_i\}_{i \in \groundset}$ be seven subspaces of a vector space
$V$ over $\F$ such that $\chr(\F) \ne 2$.  Just as in the proof Theorem~\ref{thm:fano_inequality}, we will modify the subspaces
$\{V_i\}_{ i \in \groundset}$ into $\{V'_i\}_{i \in \groundset}$ that satisfy the conditions of
Lemma~\ref{lem:conditional_nonfano_inequality}, starting with $\{V'_i\}_{i \in
  \groundset} = \{V_i\}_{i \in \groundset}$.
We again delete vectors of a certain type in an iterative fashion.
The order
 of the deletions within each step is arbitrary.

\begin{enumerate}[Step 1:]
\item Vectors in $V'_{111}$ but not in
  $\sum_{i \in \oddsm} V'_i$ from $V'_{111}$.

\item
\begin{enumerate}[(a)]
\item Vectors in $V'_{100} \cap V'_{010}$ from $V'_{010}$.
\item Vectors in $V'_{001} \cap (V'_{100} + V'_{010})$ from $V'_{001}$.
\item Vectors in $V'_{111} \cap (V'_{100} + V'_{010})$ from $V'_{111}$.
\item Vectors in $V'_{111} \cap (V'_{010} + V'_{001})$ from $V'_{111}$.
\item Vectors in $V'_{111} \cap (V'_{100} + V'_{001})$ from $V'_{111}$.
\end{enumerate}

\item Vectors in $V'_{i+j}$ but not in $V'_i+V'_j$ for $i,j \in \oddsm$ from $V'_{i+j}$.

\item Vectors in $V'_{111}$ but not in $V'_i+V'_j$ for $i,j: i+j = 111$ from $V'_{111}$.

\end{enumerate}

The first two steps in this sequence of deletions, along with the first two conditions in Lemma~\ref{lem:conditional_nonfano_inequality} are identical to those in the even characteristic case.  Thus, by arguments from the proof Theorem~\ref{thm:fano_inequality} we have that by the end of step (2) conditions (i), (ii) are satisfied.  Step (3) is almost identical to the same step in the even characteristic case; the difference is that now we only perform the step for pairs $i,j \in \oddsm$ rather than all pairs $i,j \in \odds$.  As before, at the end of step (3) condition (iii) is satisfied, and since the step does not modify $V'_i$ for any $i \in \odds$, it does not cause either of conditions (i), (ii) to become violated.  Step (4) ensures condition (iv), so it remains to show that step (4) preserves conditions (i)--(iii).  Step (4) only modifies $V'_{111}$ so it doesn't change $\sum_{i \in \oddsm} V'_i$, therefore preserving (i).  It preserves (ii) because if three subspaces are mutually complementary, they remain mutually complementary after deleting a vector from one of them.  It preserves (iii) because (iii) does not involve $V'_{111}$, which is the only subspace that changes during step (4).

Now, by Lemma~\ref{lem:conditional_fano_inequality} we have that
\begin{equation}
 3\ds(V'_{111}) \le \ds(V'_{110},V'_{101},V'_{011}).
\label{eq:uncondV'3}
\end{equation}

As in the proof of Theorem~\ref{thm:fano_inequality}, let
\begin{align*}
\delta & = \ds(V_{111}, \{V_i\}_{i \in \oddsm}) - \ds(\{V_i\}_{i \in \oddsm}) \\
\delta[i|j,k] & = \ds(V_i,V_j,V_k) - \ds(V_j,V_k) \\
\delta[i;j] & = \ds(V_i \cap V_j) = \ds(V_i) + \ds(V_j) - \ds(V_i,V_j)  \\
\delta[i;j,k] &= \ds(V_i \cap (V_j + V_k) ) =
\ds(V_i) + \ds(V_j,V_k) - \ds(V_i,V_j,V_k)
\end{align*}

Observe that we only reduce the size of subspaces, so $\ds(V'_{110},V'_{101},V'_{011}) \le \ds(V_{110},V_{101},V_{011})$.

It remains to get a lower bound on $\ds(V'_{111})$ in terms of dimensions of subsets of $\{V_i\}_{i \in \groundset}$.
We do this by giving an upper bound on the number of vectors we delete from $V'_{111}$ in terms
of the $\delta$ terms we defined above.  Step (1) deletes $\delta$ vectors.
Steps (2a) and (2b) delete nothing from $V'_{111}$, and at the end of
(2a)--(2b) we have
\begin{align}
\label{eq:nonfano-100}
\ds(V'_{100}) & = \ds(V_{100}) \\
\label{eq:nonfano-010}
\ds(V'_{010}) & = \ds(V_{010}) - \delta[100;010] \\
\label{eq:nonfano-001}
\ds(V'_{001}) & = \ds(V_{001}) - \delta[001;100,010]
\end{align}
Steps (2c)--(2e) delete at most
$\sum_{i,j \in \oddsm} \delta[111;i,j]$
vectors from $V'_{111}$, and they do not change
any of the other subspaces.

In step (3) no vectors are deleted from $V'_{111}$, but we will still need an upper bound on the number of vectors deleted in this step since it will influence our upper bound on the number of vectors deleted from $V'_{111}$ in step (4).  If no deletions took place prior to step (3), then for all $i,j \in \oddsm$ exactly $\delta[i+j|i,j]$ vectors would be deleted from $V'_{i+j}$ during step (3).  However, if $\ds(V'_i,V'_j) < \ds(V_i,V_j)$, then we must adjust our estimate of the number of deleted vectors to account for this difference.  Steps (1) and (2a) cannot change $\ds(V'_i,V'_j)$ for any $i,j \in \oddsm$, but step (2b) reduces each of $\ds(V'_{001},V'_{100})$ and $\ds(V'_{001},V'_{010})$ by at most $\delta[001;100,010]$.  Therefore, at the end of step (3) we have
\begin{align}
\label{eq:nonfano-110}
\ds(V'_{110}) &= \ds(V_{110}) - \delta[110 | 100,010] \\
\label{eq:nonfano-101}
\ds(V'_{101}) & \geq
\ds(V_{101}) - \delta[101 | 100,001] - \delta[001;100,010] \\
\label{eq:nonfano-011}
\ds(V'_{011})
& \geq \ds(V_{011}) - \delta[011 | 010,001] - \delta[001;100,010]
\end{align}

If no deletions took place prior to step (4), then the number of vectors we would need to delete from $V'_{111}$, to make it a subspace of $V'_i + V'_j$, would be at most $\delta[111 | i,j]$.  As before, we need to adjust this bound to account for the potential difference in dimension between $V_i + V_j$ and $V'_i + V'_j$.  Using the upper bound
\[
\ds(V_i + V_j) - \ds(V'_i + V'_j) \leq
\left[ \ds(V_i) - \ds(V'_i) \right] +
\left[ \ds(V_j) - \ds(V'_j) \right],
\]
which is valid for any four subspaces $V_i,V_j,V'_i,V'_j$
satisfying $V'_i \subseteq V_i, \, V'_j \subseteq V_j$, we
find that the number of extra vectors deleted from $V'_{111}$
in step (4) because of differences in dimension between
$V'_i+V'_j$ and $V_i+V_j$ (for some $i,j \in \groundset, \, i+j=111$),
is at most
\[
\sum_{i \in \groundset \setminus \{111\}}
\ds(V_i) - \ds(V'_i)
\leq
\delta[100;010] + 3 \delta[001; 100,010] +
\sum_{i,j \in \oddsm} \delta[i+j|i,j],
\]
where the first inequality follows by combining equations
\eqref{eq:nonfano-100}--\eqref{eq:nonfano-011}.

We now sum up our upper bounds on the number of vectors
deleted from $V'_{111}$ in steps (1)--(4) combined, to find
that
\begin{equation} \label{eq:nonfano-111}
\ds(V'_{111}) \geq
\ds(V_{111}) - \delta - \sum_{i,j \in \oddsm} \delta[111;i,j]
- \delta[100;010] - 3 \delta[001; 100,010] -
\sum_{i,j \in \oddsm} \delta[i+j|i,j].
\end{equation}


Expanding out all of the $\delta$ terms,
combining with the upper bound on $\ds(V'_{111})$,
and plugging these into
Equation \eqref{eq:uncondV'3} gives us
$\Lambda_{\mathrm{odd}} \cdot \dimvec(\{V_i\}_{i
  \in \groundset}) \ge 0$ for
some $2^7$-dimensional vector $\Lambda_{\mathrm{odd}}$, as desired;
after applying these steps one obtains
Equation~\eqref{eq:nonfano-full} below.
When $\{V_i\}_{i \in \groundset}$ are one-dimensional subspaces
constituting a representation of the Fano matroid
over a field of characteristic 2,
it is easy to check that all of the $\delta$ terms
appearing in~\eqref{eq:nonfano-111} are zero.  So, the
inequality states that $\ds(V_{110},V_{101},V_{011}) \ge
3\ds(V_{111}),$ whereas we know that $\ds(V_{110},V_{101},V_{011}) =
2\ds(V_{111})$ for the Fano matroid.  Consequently
$\Lambda_{\mathrm{odd}}\nolinebreak\cdot\rankvec(\fano)<0$.

For completeness, the inequality $\Lambda_{\mathrm{odd}} \cdot
\dimvec(\{V_i\}_{i \in \groundset}) \ge 0$ is written
explicitly as follows.
\begin{align} \notag
& 3 \ds(V_{100})
+ 3 \ds(V_{010})
+ 9 \ds(V_{001})
+ 6 \ds(V_{111})
+ 6 \ds(V_{100},V_{010})
- 12 \ds(V_{100},V_{010},V_{001}) \\
+ & \notag
  3 \ds(V_{110},V_{100},V_{010})
+ 3 \ds(V_{101},V_{100},V_{001})
+ 3 \ds(V_{011},V_{010},V_{001}) \\
- & \notag
  3 \ds(V_{111},V_{100},V_{010})
- 3 \ds(V_{111},V_{100},V_{001})
- 3 \ds(V_{111},V_{010},V_{001}) \\
+ & \label{eq:nonfano-full}
  3 \ds(V_{111},V_{100},V_{010},V_{001})
+   \ds(V_{110},V_{101},V_{011}) \geq 0\,.
\end{align}
This completes the proof of the theorem.
\end{proof}

\subsection{Polynomial separation between the linear and non-linear rates}

The following pair of lemmas shows how to take a single linear constraint,
such as one of those whose existence is asserted by
Theorems~\ref{thm:fano_inequality} and~\ref{thm:nonfano_inequality},
and transform it
into a tight homomorphic constraint schema.
To state the lemmas, we must first define the set of
vectors $D_{\F}(K) \subset \R^{\powerset{K}}$,
for any index set $K$ and field $\F$,
to be the set of all vectors $\dimvec(\{V_k\}_{k \in K})$,
where $\{V_k\}_{k \in K}$ runs through all $K$-indexed tuples of
finite-dimensional vector spaces over $\F$.

\begin{lemma}[Tightening Modification]\label{lem:tighten}
Suppose $I$ is any index set, $e$ is an element not in $I$,
and $J = I \cup \{e\}$.  There exists an explicit linear
transformation from $\R^{\powerset{J}}$ to $\R^{\powerset{I}}$,
represented by a matrix $B$, such that:
\begin{compactenum}[(i)]
\item \label{tighten:1}
$B \cdot D_{\F}(J) \subseteq D_{\F}(I)$ for every field $\F$.
\item \label{tighten:2}
$B \ones = B \ones_j = 0$ for all $j \in J$.
\item \label{tighten:3}
If $M$ is a matroid with ground set $I$ and the
intersection of all matroid bases of $M$ is the empty set,
then $B \rankvec(M+e) = \rankvec(M)$, where
$M + e$ denotes the
matroid obtained by adjoining a rank-zero element to $M$.
\end{compactenum}
\end{lemma}
\begin{proof}
If $U$ is any vector space with a $J$-tuple of subspaces
$\{U_j\}_{j \in J}$, then there is a quotient map $\pi$
from $U$ to $V=U/U_e$, and we can form
an $I$-tuple of subspaces $\{V_i\}_{i \in I}$ by specifying
that $V_i = \pi(U_i)$ for all $i \in I$.  The dimension
vectors $\vec{\mathbf{u}} = \dimvec(\{U_j\})$
and $\vec{\mathbf{v}} = \dimvec(\{V_i\})$ are
related by an explicit linear transformation.
In fact, for any subset $S \subseteq I$, if we
let $U_S, V_S$ denote the subspaces of $U,V$
spanned by $\{U_i\}_{i \in S}$ and $\{V_i\}_{i \in S}$,
respectively, then $\pi$ maps $U_S + U_e$ onto
$V_S$ with kernel $U_e$, and this justifies the
formula $$\mathbf{v}_S = \mathbf{u}_{S \cup \{e\}} -
\mathbf{u}_{\{e\}}.$$  Thus, $\mathbf{v} = B_0 \mathbf{u}$,
where $B_0$ is the matrix
\begin{equation} \label{eq:b0}
(B_0)_{ST} = \begin{cases}
1 & \mbox{if $T = S \cup \{e\}$} \\
-1 & \mbox{if $T = \{e\}$} \\
0 & \mbox{otherwise},
\end{cases}
\end{equation}
and therefore $B_0 \cdot D_{\F}(J) \subseteq
D_{\F}(I)$.

Similarly, if $U$ is any vector space with an $I$-tuple
of subspaces $\{U_i\}_{i \in I}$ and $k$ is any element of $I$,
we can define $U_{-k} \subseteq U$ to be the linear
subspace spanned by $\{U_i\}_{i \neq k}$, and we can
let $\pi : U \to U_{-k}$ be any linear transformation
whose restriction to $U_{-k}$ is the identity.  The
restriction of $\pi$ to $U_k$ has kernel $W_k$ of
dimension $\dim(W_k) = \dim(\{U_i\}_{i \in I}) -
\dim(\{U_i\}_{i \in I, i \neq k})$.  As
before, let $V_i = \pi(U_i)$ for all $i \in I$,
let $U_S,V_S$ denote the subspaces of $U,V$
spanned by $\{U_i\}_{i \in S}$ and $\{V_i\}_{i \in S}$,
and let $\vec{\mathbf{u}} = \dimvec(\{U_i\}),
\vec{\mathbf{v}} = \dimvec(\{V_i\})$.
If $k \not\in S$ then $V_S = U_S$
and $\mathbf{v}_S = \mathbf{u}_S$,
while if $k \in S$ then $U_S$ contains $W_k$, the
linear transformation $\pi$ maps $U_S$ onto $V_S$
with kernel $W_k$, and
$\mathbf{v}_S = \mathbf{u}_S - \dim(W_k) =
\mathbf{u}_S - \mathbf{u}_I + \mathbf{u}_{I \setminus \{k\}}$.
Thus, $\mathbf{v} = B_k \mathbf{u}$, where $B_k$ is
the matrix
\begin{equation} \label{eq:bk}
(B_k)_{ST} = \begin{cases}
1 & \mbox{if $T = S$} \\
1 & \mbox{if $k \in S$ and $T = I \setminus \{k\}$} \\
-1 & \mbox{if $k \in S$ and $T = I$} \\
0 & \mbox{otherwise}.
\end{cases}
\end{equation}
and therefore $B_k \cdot D_{\F}(I) \subseteq
D_{\F}(I)$.

Now assume without loss of generality that
$I = \{1,2,\ldots,n\}$ and let $B = B_n B_{n-1} \cdots B_1 B_0$.
We have seen that $B \cdot D_{\F}(J) \subseteq D_{\F}(I)$.
From~\eqref{eq:b0}
one can see that $B_0 \ones = B_0 \ones_e = 0$
and that for every $k \in I$, $B_0 \ones_k = \ones_k$.
(Here, it is important to note that $\ones_k$
on the left side refers to a vector in $\R^{\powerset{J}}$
and on the right side it refers to a vector in $\R^{\powerset{I}}$.)
Furthermore, from~\eqref{eq:bk} one can see that
$B_k \ones_k = 0$ and that $B_k \ones_i = \ones_i$
for all $i \neq k$.  Thus, when we left-multiply
a vector $\vec{\mathbf{w}} \in \{\ones\} \cup
\{\ones_j\}_{j \in J}$ by the matrix $B$, one of
the following things happens.  If $\vec{\mathbf{w}}$ is
equal to $\ones$ or $\ones_e$ then $B_0 \vec{\mathbf{w}}=0$
hence $B \vec{\mathbf{w}} = 0$.  Otherwise,
$\vec{\mathbf{w}} = \ones_k \in \R^{\powerset{J}}$
for some $k \in I$,
$B_0 \vec{\mathbf{w}} = \ones_k \in \R^{\powerset{I}}$,
and as we proceed to left-multiply $\ones_k$ by
$B_1, B_2, \ldots,$ it is fixed by $B_i \, (i < k)$
and annihilated by $B_k$, so once again $B \vec{\mathbf{w}} = 0$.
This confirms assertion~(\ref{tighten:2}) of the lemma.

Finally, if $M, M+e$ are matroids satisfying the hypotheses
of assertion~(\ref{tighten:3}), then for every set $S \subseteq I$
we have $r(S \cup \{e\}) - r(\{e\}) = r(S)$ and hence
$B_0 \rankvec(M+e) = \rankvec(M)$.  For any \mbox{$k \in I$}
our assumption on $M$ implies that it has a matroid basis
disjoint from $\{k\}$, and hence that \mbox{$r(I \setminus \{k\}) = r(I)$}.
Inspecting~\eqref{eq:bk}, we see that this implies $B_k
\rankvec(M) = \rankvec(M)$ for all $k \in I$, and hence
$B \rankvec(M+e) = \rankvec(M)$ as desired.
\end{proof}

\begin{lemma}[Homomorphic Schema Extension] \label{lem:homo}
Let $I$ be an index set, and
let $\vec{\alpha} \in \R^{\powerset{I}}$ be a
vector such that $\vec{\alpha}^\trans \dimvec \ge 0$
for all $\dimvec \in D_{\F}(I)$.
Then there is a homomorphic constraint schema $(Q,\cm)$ such that
$\vec{\alpha}^{\trans}$ is a row of the matrix $\cm(I)$, and
for every index set $K$ and vector $\dimvec \in D_\F(K)$,
$\cm(K) \dimvec \geq 0$.  If $\vec{\alpha}^\trans \ones = \vec{\alpha}^\trans
\onesi = 0$ for all $i \in I$, then the constraint schema
$(Q,\cm)$ is tight.
\end{lemma}
\begin{proof}
For any index set $J$, let $Q(J)$ be the set of all
Boolean lattice homomorphisms from $\powerset{I}$ to $\powerset{J}$.  If
$h:\powerset{J} \to \powerset{K}$ is another Boolean lattice homomorphism,
the mapping $h_*$ is defined by function composition,
i.e.\ $h_*(q) = h \circ q$.

To define the constraint matrix $\cm(J)$ associated
to an index set $J$, we do the following.  A row of
$\cm(J)$ is indexed by a 
Boolean lattice homomorphism
$q : \powerset{I} \to \powerset{J}$ and we define the entries of that row
by
\begin{equation} \label{eq:homo-ext}
\cm(J)_{q S} =
\sum_{\substack{T \in \powerset{I} \\ q(T) = S}} \alpha_T.
\end{equation}
This defines a homomorphic constraint schema, because
if $h : \powerset{J} \to \powerset{K}$ is any Boolean
lattice homomorphism and $R = \cm(K)^\trans Q_h, \,
R' = P_h \cm(J)^\trans$, then recalling the definitions
of $P_h,Q_h$ we find that
\begin{align*}
R_{Sq} & = (\cm(K)^\trans)_{S,h_*(q)} = \cm(K)_{h \circ q,S} =
\sum_{\substack{T \in \powerset{I} \\ h(q(T))=S}} \alpha_T \\
R'_{Sq} &= \sum_{S' : h(S')=S} (\cm(J)^{\trans})_{S',q} =
\sum_{S' : h(S')=S}
\sum_{\substack{T \in \powerset{I} \\ q(T)=S'}} \alpha_T
\end{align*}
and the right-hand sides of the two lines are clearly
equal.

To prove that $\cm(K) \dimvec \geq 0$ for every
$\dimvec \in D_\F(K)$, we reason as follows.  It
suffices to take a single row of the constraint
matrix, indexed by homomorphism $q : I \rightarrow K$,
and to prove that
$$
\sum_{S \in \powerset{K}} \cm(K)_{q S} \dimvec_S \geq 0.
$$
Using the definition of the constraint matrix entries,
this can be rewritten as
\begin{equation} \label{eq:foo}
\sum_{T \in \powerset{I}} \alpha_T \dimvec_{q(T)} \geq 0.
\end{equation}
Let $\{V_k\}_{k \in K}$ be a $K$-tuple of vector
spaces such that $\dimvec = \dimvec(\{V_k\}_{k \in K})$.
Define an $I$-tuple of vector spaces
$\{U_i\}_{i \in I}$ by setting
$U_i$ to be the span of $\{V_k\}_{k \in q(\{i\})}$.
By our hypothesis on $\vec{\alpha}$,
$$
\sum_{T \in \powerset{I}} \alpha_T \dim(\{U_i\}_{i \in T}) \geq 0.
$$
The left side is equal to the left side of~\eqref{eq:foo}.

Finally, suppose that $\vec{\alpha}^\trans \ones = \vec{\alpha}^\trans \onesi$
for all $i \in I$.  For any index set $J$,
we prove that $\cm(J) \ones = 0$ by calculating the component
of $\cm(J) \ones$ indexed by an arbitrary Boolean lattice
homomorphism $q : \powerset{I} \to \powerset{J}$.
\begin{align*}
\sum_{S \in \powerset{J}} \cm(J)_{qS} &=
\sum_{S \in \powerset{J}} \sum_{\substack{T \in \powerset{I} \\ q(T)=S}} \alpha_T =
\sum_{T \in \powerset{I}} \alpha_T = \vec{\alpha}^\trans \ones = 0
\end{align*}
The proof that $\cm(J) \ones_j = 0$ for all $j \in J$ similarly
calculates the component of $\cm(J) \ones_j$ indexed by an arbitrary $q$.
\begin{align*}
\sum_{\substack{S \in \powerset{J} \\ j \in S}} \cm(J)_{qS} &=
\sum_{\substack{S \in \powerset{J} \\ j \in S}}
\sum_{\substack{T \in \powerset{I} \\ q(T)=S}} \alpha_T =
\sum_{\substack{T \in \powerset{I} \\ j \in q(T)}} \alpha_T
\end{align*}
At this point the argument splits into three cases.
If $j \in q(\emptyset)$ then the right side is
$\vec{\alpha}^\trans \ones$, which equals 0.  If
$j \not\in q(J)$ then the right side is an empty sum
and clearly equals 0.  If $j \not\in q(\emptyset)$
but $j \in q(J)$, then there is a unique $i \in I$
such that $j \in q(\{i\})$.  Indeed, if $j$ belongs to
$q(\{i\})$ and $q(\{i'\})$, then $j$ belongs to
$q(\{i\}) \cap q(\{i'\}) = q(\{i\} \cap \{i'\})$,
implying that $\{i\} \cap \{i'\}$ is non-empty and
that $i=i'$.  The right side of the equation above
is thus equal to $\vec{\alpha}^\trans \onesi$, which equals 0.
\end{proof}


Finally, before proving Theorem~\ref{thm:separation}, it
will be useful to describe the following simple operation
for combining constraint schemas.
\begin{definition} \label{def:schema-sum}
The \emph{disjoint union} of two constraint schemas
$(Q_1,A_1)$ and $(Q_2,A_2)$ is the constraint schema
which associates to every index set $I$
the disjoint union $\cset(I) = \cset_1(I) \sqcup
\cset_2(I)$ and the constraint matrix $\cm(I)$ given by
$$
\cm(I)_{qS} = \begin{cases} \cm_1(I)_{qS} & \mbox{if $q \in \cset_1(I)$} \\
\cm_2(I)_{qS} & \mbox{if $q \in \cset_2(I)$.}
\end{cases}
$$
For a homomorphism $h : \powerset{I} \to \powerset{J}$,
the function $h_* : \cset_1(I) \sqcup \cset_2(I) \to \cset_1(J) \sqcup
\cset_2(J)$ is defined by combining
$\cset_1(I) \stackrel{h_*}{\longrightarrow} \cset_1(J)$ and
$\cset_2(I) \stackrel{h_*}{\longrightarrow} \cset_2(J)$ in the obvious way.
\end{definition}

\begin{lemma} \label{lem:schema-sum}
The disjoint union of two tight constraint schemas is tight,
and the disjoint union of two homomorphic constraint schemas
is homomorphic.
\end{lemma}
\begin{proof}
For all index sets $I$ and vectors $v \in \R^{\powerset{I}}$,
the constraint matrix of the disjoint union satisfies
$$
\cm(I) \mathbf{v} =
\begin{pmatrix} \cm_1(I) \\ \cm_2(I) \end{pmatrix} \mathbf{v} =
\begin{pmatrix} \cm_1(I) \mathbf{v} \\ \cm_2(I) \mathbf{v} \end{pmatrix}
$$
so if both constraint schemas are tight then so is their disjoint union.
If $h : \powerset{I} \to \powerset{J}$ is a Boolean lattice homomorphism
then let $Q_{1h}, Q_{2h}, Q_h$ denote the matrices representing the
induced linear transformations $\R^{\cset_1(I)} \to \R^{\cset_1(J)}$,
$\R^{\cset_2(I)} \to \R^{\cset_2(J)}$, and
$\R^{\cset(I)} \to \R^{\cset(J)}$, respectively.
If both constraint schemas are homomorphic, then
\begin{align*}
\cm(J)^\trans Q_h =
\begin{pmatrix} \cm_1(J)^\trans & \cm_2(J)^{\trans} \end{pmatrix}
\begin{pmatrix} Q_{1h} & 0 \\ 0 & Q_{2h} \end{pmatrix}
& =
\begin{pmatrix} \cm_1(J)^{\trans} Q_{1h} &
\cm_2(J)^\trans Q_{2h} \end{pmatrix}
\\
& =
\begin{pmatrix}
P_h \cm_1(J)^\trans & P_h \cm_2(J)^\trans
\end{pmatrix}
=
P_h \cm(J)^\trans,
\end{align*}
which confirms that the disjoint union is homomorphic.
\end{proof}

The proof of Theorem~\ref{thm:separation} now follows
by combining earlier results.
\begin{proof}[\textbf{\emph{Proof of Theorem~\ref{thm:separation} (Separation Theorem)}}]
The fact that $\beta(G_{\fano}) = \beta(G_{\nonfano}) = 4$ is an
immediate consequence of Theorem~\ref{thm:beta_rep_matroids}.
The submultiplicativity of $\beta$ under the lexicographic
product (Theorem~\ref{thm:beta-submult}) then implies
that $G = G_{\fano} \lexp G_{\nonfano}$ satisfies $\beta(G^{\lexp n}) \le (4\cdot 4)^n = 16^n$.
A lower bound of the form $\beta(G^{\lexp n}) \geq 16^n$ is a consequence of
Proposition~\ref{prop:matroidLP}
which implied that $b(G_\fano) = b(G_\nonfano) = 4$, from which it follows by the supermultiplicativity of $b$ under
lexicographic products (Theorem~\ref{thm:lexprod})
that $16^n \leq b(G^{\lexp n}) \leq \beta(G^{\lexp n})$.
Combining these upper and lower bounds, we find that
$\beta(G^{\lexp n}) = 16^n$.

It is worth noting, incidentally, that although each of
$G_\fano, \, G_\nonfano$ individually has a linear solution
over the appropriate field, the index code for $G = G_{\fano} \lexp
G_{\nonfano}$ implied by the
proof of Theorem~\ref{thm:beta-submult} --- which concatenates
these two linear codes together by composing them with an
arbitrary one-to-one mapping from a mod-2 vector space to a
mod-$p$ vector space ($p$ odd) --- is highly nonlinear,
and not merely a side-by-side application of two linear codes.

To establish the lower bound on $\linrate^{\F}(G^{\lexp n})$,
we distinguish two cases, $\chr(\F) = 2$ and
$\chr(\F) \ne 2$, and
in both cases we prove $\linrate^{\F}(G^{\lexp n}) \geq (16+\eps)^n$
using the LP in Figure~\ref{fig:lp} with a tight homomorphic constraint schema
supplying the constraint matrix $A$.  We
use different constraint schemas in the two cases
but the constructions are nearly identical.
Let $M$ denote the matroid $\nonfano$ if $\chr(\F)=2$, and let
$M=\fano$ if $\chr(\F) \neq 2$.  In both cases, we will let
$M+e$ denote the matroid obtained by adjoining a rank-zero element
to $M$, and we will denote
the ground sets of $M, \, M+e$ by $I,J$, respectively.
Recall the vectors $\Lambda_{\mathrm{even}}, \Lambda_{\mathrm{odd}} \in
\R^{\powerset{I}}$ from
Theorems~\ref{thm:fano_inequality} and~\ref{thm:nonfano_inequality}.
Let $\Lambda = \Lambda_{\mathrm{even}}$ if $\chr(\F)=2$,
$\Lambda = \Lambda_{\mathrm{odd}}$ if $\chr(\F) \neq 2$.
By Theorems~\ref{thm:fano_inequality} and~\ref{thm:nonfano_inequality},
$\Lambda \cdot \rankvec(M) < 0$, a fact that we will be using later.

Recall the linear
transformation $B : \R^{\powerset{J}} \to \R^{\powerset{I}}$
from Lemma~\ref{lem:tighten}, and let
$$\vec{\alpha} = B^\trans \Lambda.$$
For any $\dimvec \in D_{\F}(J)$ we have
$\vec{\alpha}^\trans \dimvec = \Lambda^\trans B \dimvec \geq 0$,
since $B \dimvec \in D_{\F}(I)$ and
$\Lambda \cdot \vec{\mathbf{v}} \geq 0$ for
all $\vec{\mathbf{v}} \in D_{\F}(I)$.
The equations $B \ones = B \ones_j = 0$ for all $j \in J$
imply that $\vec{\alpha}^\trans \ones = \vec{\alpha}^\trans \ones_j = 0$.
Applying Lemma~\ref{lem:homo} to obtain a tight homomorphic constraint schema
from $\vec{\alpha}$, and taking its disjoint union with
the submodularity constraint schema, we arrive at
a tight homomorphic constraint schema $(Q,A)$ such that every
vector $\dimvec \in D_{\F}(K)$, for every index set $K$, satisfies
the system of inequalities $A(K) \dimvec \geq 0$.

Consider the LP in Figure~\ref{fig:lp}, instantiated with constraint
schema $(Q,A)$.  We claim that its optimal solution $\bsf(G)$,
for any index coding problem $G$,
satisfies $\bsf(G) \leq \linrate^\F(G)$.
To prove this we proceed as in the proof
of Theorem~\ref{thm:lplowerbound}: consider
any linear index code over $\F$ with message alphabet $\Sigma$,
sample each message independently
and uniformly at random, and consider the input messages and the
broadcast message as random variables, with $H(S)$ denoting the joint entropy
of a subset $S$ of these random variables.  The set of random
variables is indexed by $K = V(G) \cup \{e\}$, where $V(G)$
denotes the set of messages in $G$ and $e$ is an extra
element of $K$ corresponding to the broadcast message.
Letting $M_k$, for $k \in K$, denote the matrix representing
the linear transformation defined by the $k^{\mathrm{th}}$
random variable, and letting $U_k$ denote the row space
of $M_k$, we have $H(S) = \log |\F| \cdot \dim(\{U_k\}_{k \in S})$
for every $S \subseteq K$.
For $S \subseteq V(G)$ let
$$z_S = \frac{H(S \cup \{e\})}{\log |\Sigma|}
= \left(
\frac{\log |\F|}{\log |\Sigma|}
\right)
\dim(U_e, \, \{U_i\}_{i \in S}).
$$
We aim to show that $z$ satisfies the constraints
of the LP, implying that $\bsf(G) \leq z_{\emptyset}
= \frac{\log |\Sigma_P|}{\log|\Sigma|}$
and consequently
(since the linear index code over $\F$  was arbitrary)
that $\bsf(G) \leq \linrate^\F(G)$.
The proof of Theorem~\ref{thm:lplowerbound} already
established that $z_{V(G)} = |V(G)|$ and that
$z_T - z_S \leq c_{ST}$ for all $S \subset T$,
so we need only show that $A z \geq 0$.
As in the proof of Lemma~\ref{lem:tighten},
we let $\pi$ denote the quotient map from $U = \Sigma^{V(G)}$
to $U / U_e$,  we define $V_i = \pi(U_i)$,  and we
observe that for all $S \subseteq V(G)$,
$\dim(\{V_i\}_{i \in S}) = \dim(U_e, \, \{U_i\}_{i \in S}) - \dim(U_e).$
This implies that $$z - z_{\emptyset} \ones = \left(
\frac{\log |\F|}{\log |\Sigma|} \right)
\dimvec,$$
where $\dimvec = \dimvec(\{V_i\}).$
Our construction of $A$ implies that $A \dimvec \geq 0$
and that $A \ones = 0$, hence $A z \geq 0$.

It remains to show that $\bsf(G) > 16$, from
which the claimed lower bound follows by supermultiplicativity.
Since our constraint schema includes submodularity,
we have $\bsf(G_\fano) \geq b(G_\fano) = 4$
and $\bsf(G_\nonfano) \geq b(G_\nonfano) = 4$,
so we need only show that one of these inequalities
is strict, and we accomplish this using Theorem~\ref{thm:addIneq}.
Specifically, we show that the matrix $A = \cm(I)$
has a row indexed by some $q \in \cset(I)$,
such that $\left(A\, \rankvec(M)\right)_q < 0$.
Recall that $\cset(I)$ is the set of Boolean lattice
homomorphisms from $\powerset{J}$ to $\powerset{I}$,
where $J = I \cup \{e\}$.
Let $q$ be the homomorphism that maps $\{e\}$ to $\emptyset$
and $\{i\}$ to itself for every $i \in I$.  To prove
that $\left(A\, \rankvec(M)\right)_q < 0$, we let
$r(\cdot)$ and $\hat{r}(\cdot)$ denote the rank functions
of $M, \, M+e$, respectively, and we recall the definition
of the matrix entries $a_{qS}$ from~\eqref{eq:homo-ext},
to justify the following calculation:
\begin{align*}
\left(A\, \rankvec(M)\right)_q = \sum_{S \in \powerset{I}} a_{qS} \, r(S)
=
\sum_{T \in \powerset{J}} \alpha_T \, r(q(T))
& =
\sum_{T \in \powerset{J}} \alpha_T \, \hat{r}(T) \\
& =
\vec{\alpha}^\trans \rankvec(M+e)
=
\Lambda^\trans B \rankvec(M+e)
=
\Lambda^\trans \rankvec(M)
< 0
\end{align*}
where the last two steps used Lemma~\ref{lem:tighten}(\ref{tighten:3})
and Theorem~\ref{thm:fano_inequality} or~\ref{thm:nonfano_inequality} for a field $\F$ of even or odd characteristic, respectively.
\end{proof}

\end{document}